\documentclass[11pt]{article}

\usepackage{fullpage}

\usepackage{amsmath,amssymb,amsfonts,mathrsfs, mathtools, bm, bbm, dsfont, mathrsfs, amsthm}
\usepackage{graphicx, epstopdf}

\makeatletter
\newcommand{\removelatexerror}{\let\@latex@error\@gobble}
\makeatother

\usepackage{ifthen}
\usepackage{multirow, multicol}
\usepackage{float,framed}
\usepackage{subeqnarray, array}
\usepackage{enumitem}
\usepackage[dvipsnames]{xcolor}
\usepackage{caption}
\usepackage[caption=false, font=small]{subfig}
\usepackage[hyphens]{url}
\usepackage{cite}
\usepackage{booktabs}

\usepackage{accents}

\usepackage[breaklinks, colorlinks, linkcolor=MidnightBlue, anchorcolor=MidnightBlue, citecolor=MidnightBlue, urlcolor=MidnightBlue]{hyperref}

\newcommand{\beq}{\begin{equation}}
\newcommand{\eeq}{\end{equation}}
\newcommand{\beqa}{\begin{eqnarray}}
\newcommand{\eeqa}{\end{eqnarray}}
\newcommand{\beqan}{\begin{eqnarray*}}
\newcommand{\eeqan}{\end{eqnarray*}}

\newcommand{\rank}{\text{rank }}

\newcommand{\trace}{\mbox{\rm Tr}}

\newcommand{\dom}{\textrm{dom }}

\newcommand\T{{\mathpalette\raiseT\intercal}}
\newcommand\raiseT[2]{\raisebox{0.25ex}{$#1#2$}
}

\newcommand{\epi}{{\textrm{epi }}}

\newcommand{\Aset}{\mathds{A}}

\newcommand{\Cset}{\mathds{C}}

\newcommand{\Eset}{\mathds{E}}

\newcommand{\Gset}{\mathds{G}}
\newcommand{\Hset}{\mathds{H}}

\newcommand{\Mset}{\mathds{M}}
\newcommand{\Nset}{\mathds{N}}

\newcommand{\Rset}{\mathds{R}}

\newcommand{\Wset}{\mathds{W}}
\newcommand{\Xset}{\mathds{X}}

\newcommand{\Dcal}{{\cal D}}

\newcommand{\Lcal}{{\cal L}}

\newcommand{\Pcal}{{\cal P}}

\newcommand{\Vcal}{{\cal V}}

\newcommand{\Hsf}{\sf{H}}

\newcommand{\bone}{\mathbf{1}}

\renewcommand{\v}[1]{{\bm{#1}}}

\newcommand{\ol}[1]{\ensuremath{\overline{{#1}}}}
\newcommand{\ul}[1]{\ensuremath{\underline{{#1}}}}

\newcounter{l1}
\newcounter{l2}
\newcounter{l3}
\newcommand{\bdotlist}{\begin{list}{$\bullet$}{}}
\newcommand{\bboxlist}{\begin{list}{$\Box$}{}}
\newcommand{\bbboxlist}{\begin{list}{\raisebox{.005in}{{\tiny
$\blacksquare$ \ \ }}}{}}
\newcommand{\bdashlist}{\begin{list}{$-$}{} }
\newcommand{\blist}{\begin{list}{}{} }
\newcommand{\barablist}{\begin{list}{\arabic{l1}}{\usecounter{l1}}}
\newcommand{\balphlist}{\begin{list}{(\alph{l2})}{\usecounter{l2}}}
\newcommand{\bAlphlist}{\begin{list}{\Alph{l2}.}{\usecounter{l2}}}
\newcommand{\bdiamlist}{\begin{list}{$\diamond$}{}}
\newcommand{\bromalist}{\begin{list}{(\roman{l3})}{\usecounter{l3}}}

\newtheorem{theorem}{Theorem}
\newtheorem{lemma}{Lemma}
\newtheorem{proposition}{Proposition}
\newtheorem{corollary}{Corollary}

\newtheorem{definition}{Definition}

\renewcommand{\bone}{\mathds{1}}
\newcommand{\AC}{{\sf AC}}
\newcommand{\SDP}{{\sf{SDP}}}

\newcommand{\Pac}{{\Pcal_{\AC}}}
\newcommand{\Psdp}{{\Pcal_{\SDP}}}
\newcommand{\SO}{{\sf SO}}
\newcommand{\opt}{{\sf opt}}
\newcommand{\Psocp}{{\Pcal_{\sf SOCP}}}

\newcommand{\conv}{\textrm{conv }}
\newcommand{\fc}{\textrm{c}}

\newcommand{\bose}[1]{
{ \textcolor{red}{(Bose says:  #1)}}{}}

\usepackage{palatino}

\begin{document}

\title{Pricing Economic Dispatch with AC Power Flow via Local Multipliers and Conic Relaxation}
\author{Mariola Ndrio \qquad Anna Winnicki   \qquad Subhonmesh Bose\thanks{All authors are with the Department of Electrical and Computer Engineering at the University of Illinois at Urbana-Champaign, Urbana, IL 61801. Emails: \{ ndrio2, annaw5, boses\}@illinois.edu. 
		This work was partially supported by grants from the Power Systems Engineering Research Center (PSERC) and by the National Science Foundation under grant no. CAREER-2048065. 
	}}

	
	\maketitle

\begin{abstract}
We analyze pricing mechanisms in electricity markets with AC power flow equations that define a nonconvex feasible set for the economic dispatch problem. Specifically, we consider two possible pricing schemes. The first among these prices are derived from Lagrange multipliers that satisfy Karush-Kuhn-Tucker conditions for local optimality of the nonconvex market clearing problem. The second is derived from optimal dual multipliers of the convex semidefinite programming (SDP) based relaxation of the market clearing problem. Relationships between these prices, their revenue adequacy and market equilibrium properties are derived and compared. 
The SDP prices are shown to equal distribution locational marginal prices derived with second-order conic relaxations of power flow equations over radial distribution networks. We illustrate our theoretical findings through numerical experiments.
\end{abstract}

\section{Introduction}
Electricity markets rely on a bid-based security-constrained economic dispatch (ED) problem to compute dispatch and pricing decisions. Prices are derived as optimal dual multipliers of system constraints in the ED problem. These locational marginal prices (LMPs), proposed in \cite{schweppe1984}, reflect marginal system costs to meet local incremental demand. 
LMPs enjoy strong theoretical guarantees when the market clearing problem is convex, e.g., when derived with lossless linearized power flow models and convex generation costs. In this paper, we analyze price formation that accounts for nonconvexity in the market clearing problem. Nonconvexity can arise from two sources--unit commitment considerations with startup/no-load costs and the alternating current (AC) power flow equations. Pricing with binary commitment decisions have been extensively studied, e.g., see \cite{ONEILL,Hogan2003,Vazquez,Hua,Zhao,Gribik}. In this paper, we focus on price formation with nonconvexities that arise from an AC power flow model, which has received much less attention (see \cite{Garcia} for recent work). This nonconvexity is not a consequence of the cost structures of assets, but rather stems from the nature of the Kirchhoff's laws that govern the underlying power network. There is an increasing interest to efficiently and optimally solve the nonconvex market clearing problem with AC power flow, e.g., the ongoing ARPA-E GO competition. We are motivated to analyze meaningful prices that can accompany such a dispatch.

We consider two candidate pricing mechanisms that we analyze and compare in this paper. The first pricing scheme utilizes Lagrange multipliers obtained from a nonconvex dispatch problem that satisfy  Karush-Kuhn-Tucker (KKT)  conditions that are necessary for locally optimal dispatch solutions, much along the lines of \cite{Garcia}. We call these prices AC-LMPs. The second pricing mechanism utilizes optimal Lagrange multipliers from a semidefinite programming (SDP)-based convex relaxation of the economic dispatch problem as prices. We call these prices SDP-LMPs. The second pricing scheme is motivated by recent work on distribution locational marginal prices derived from a second-order cone programming (SOCP)-based convex relaxation of power flow equations in radial distribution networks (see \cite{PapavasiliouDLMPs, Caramanis, Yuan2018}). We call the latter SOCP-DLMPs in the sequel. In short, our work seeks to understand and compare economically relevant properties of AC-LMPs and SDP-LMPs, where the latter can be viewed as a generalization of SOCP-DLMPs.

The ED problem with AC power flow equations is nonconvex. We show that when this problem has zero duality gap, AC-LMPs associated with global optimal solutions and SDP-LMPs coincide. With non-zero duality gap, however, these prices can be different. Moreover, AC-LMPs associated with local--but not global--minima can be different from SDP-LMPs as well. Our derivation of this result exploits the fact that the nonconvex ED problem and its SDP relaxation share the same Lagrangian dual  program. When duality gap vanishes, the SDP relaxation essentially provides a globally optimal solution to the ED problem. Not surprisingly, the prices coincide as well. Such a relationship does not exist between the optimal solutions of the SDP relaxation and local optima or that with global optima with nonzero duality gap. As a result, these prices can be different.

The nonconvex ED problem and the SDP relaxation are dual equivalent optimization problems. Consequently, SDP-LMPs can be derived as dual optimizers of the ED problem. Note that convex hull pricing (CHP) in \cite{Gribik, Schiro} also advocates pricing via the Lagrangian dual program of a market clearing problem with nonconvex cost structures. Even though SDP-LMPs and CHPs have similar roots and their properties appear similar on the surface, there are interesting differences. For example, the authors of \cite{Gribik} study the variation of the optimal cost of the nonconvex market clearing problem to nodal power demands. They show that the convex hull of the epigraph of this optimal cost is in fact the epigraph of the optimal cost of its Lagrangian dual problem from which the CHPs are derived. One might surmise that an identical relationship holds between the optimal cost variation of the ED problem with AC power flow equations and its SDP relaxation with nodal demands. We show, however, that the argument in \cite{Gribik} breaks in our setting!

A pricing mechanism is revenue adequate when the revenue collected by the system operator (SO) from consumers is enough to cover the rents payable to suppliers. We derive a sufficient condition for revenue adequacy with AC-LMPs (and SDP-LMPs under zero duality gap). We illustrate through an example that our condition we identify is sufficient, but not necessary, for revenue adequacy.

We demonstrate that AC-LMPs always support a market equilibrium, i.e., they adequately incentivize all market participants to follow the SO-prescribed dispatch. This assertion holds even at a local optimal solution of the nonconvex ED problem.
SDP-LMPs, on the other hand, only support market equilibrium when the duality gap of the ED problem is zero. With non-zero duality gap, pricing via SDP-LMPs may require side-payments from the SO to incentivize them to follow the SO-prescribed dispatch. We characterize the duality gap of the ED problem as the minimization of two terms, the first among which is the aggregate side-payments, very much along the lines of \cite{Schiro} for CHP. However, we argue why such a formula for the duality gap does not make SDP-LMPs a minimizer of side-payments. If minimization of said payments is the ultimate goal, AC-LMPs suffice. Again, this result highlights the subtle differences between SDP-LMPs and CHPs, even though both these pricing mechanisms are derived from the Lagrangian dual of the nonconvex market clearing formulations with two different types of nonconvexities.

Finally, we prove that SDP-LMPs reduce to SOCP-DLMPs in \cite{PapavasiliouDLMPs} over radial distribution networks. This observation leverages known results in \cite{bose2015equivalent, madani2014convex} that the SDP relaxation of the ED problem over radial networks can be solved as an SOCP. Our analysis of SDP-LMPs therefore directly provides insights into revenue adequacy, market equilibrium, sensitivity of prices to demand and necessary side-payments for market participants with SOCP-DLMPs. Such an analysis is particularly timely, given recent interests in the design of retail markets with DLMPs (e.g., see  \cite{Ntakou, PapavasiliouDLMPs, Yuan2018}) to harness the flexibility of distributed energy resources at the grid-edge.

The paper is organized as follows.
In Section \ref{sec:market}, we define the ED problem with AC power flow. Then in Section \ref{sec:prices}, we introduce the two pricing mechanisms (AC-LMPs and SDP-LMPs) and establish the relationship between them. We study revenue adequacy of the pricing schemes in Section \ref{sec:ra} and market equilibrium properties in Section \ref{sec:marketEq}. We study these prices on a three-bus power network example in Section \ref{sec:3bus}. In Section \ref{sec:DLMP}, we establish the connection between SDP-LMPs and SOCP-based DLMPs. We 
conclude in Section \ref{sec:conc}.

\section{Economic Dispatch with AC Power Flow}\label{sec:market}


Consider an electric power network on $n$ buses and $m$ transmission lines. Let $\v{V} \in \Cset^n$ denote the vector of nodal voltage phasors, where $\Cset$ is the set of complex numbers. Denote by $y_{k\ell}$, the admittance of the line joining buses $k, \ell$.
The apparent power flow from bus $k$ to bus $\ell$ is
\begin{align}
\begin{aligned}
p_{k\ell} + \bm{i} q_{k\ell} 
= \v{V}^{\Hsf}  \v{\Phi}_{k \ell}  \v{V} 
+ \bm{i} \v{V}^{\Hsf}  \v{\Psi}_{k\ell}  \v{V},
\end{aligned}
\label{eq:pkl.qkl}
\end{align}
where $\v{\Phi}_{k\ell}, \v{\Psi}_{k\ell} \in \Hset^n$ that comprises all zeros except 
\begin{gather*}
[\v{\Phi}_{k \ell}]_{kk} := \frac{1}{2}(y_{k\ell} + y_{k\ell}^{\Hsf}), \ \ 
[\v{\Phi}_{k \ell}]_{k \ell} =  [\v{\Phi}_{k \ell}]_{\ell k}^{\Hsf} := -\frac{1}{2}y_{k\ell}, 
\\
[\v{\Psi}_{k \ell}]_{kk} := \frac{1}{2 \bm{i}}(y_{k\ell}^{\Hsf} - y_{k\ell}), \ \ 
[\v{\Psi}_{k \ell}]_{k \ell} :=  [\v{\Psi}_{k \ell}]_{\ell k}^{\Hsf} := \frac{1}{2\bm{i}}y_{k\ell}.
\end{gather*}
Here, $\v{A}^{\Hsf}$ calculates the conjugate transpose of an arbitrary matrix $\v{A}$, $\bm{i} := \sqrt{-1}$ and $\Hset^n \subset \Cset^{n\times n}$ is the set of Hermitian matrices.
The two summands in the right-hand-side of \eqref{eq:pkl.qkl} define the real and reactive power flows from bus $k$ to bus $\ell$, respectively.  Assume that the real power flows on the lines are constrained as
\begin{align}
p_{kl} \leq f_{k\ell} 
\label{eq:line.const}
\end{align}
for a flow limit $f_{k\ell} > 0$. Such limits typically arise from thermal considerations, but may also serve as proxies for stability constraints.\footnote{Line flow constraints are often formulated over the apparent power flow as $p_{k\ell}^2 + q_{k\ell}^2 \leq f^2_{k\ell}$ that constrain the magnitude of the current flowing over the transmission line. We consider limits on real power flow for simplicity.}
Assume that $y_{kk}$ is the shunt admittance at bus $k$. Then, the apparent power injection at bus $k$ becomes 
\begin{align*}
p_k + \bm{i} q_k  
=\v{V}^{\Hsf}  \v{\Phi}_k \v{V} + \bm{i} \v{V}^{\Hsf}  \v{\Psi}_k \v{V},
\end{align*}
where 
\begin{align}
{
	\begin{aligned}
	\v{\Phi}_k 
	&:=
	\frac{1}{2}\left( {y_{kk} + y_{kk}^{\Hsf}}\right) \bone_k \bone_k^{\Hsf} 
	+ \sum_{\ell \sim k} \v{\Phi}_{k\ell}, 
	\qquad
	\v{\Psi}_k 
	&:= 
	\frac{1}{2\bm{i}} \left( {y_{kk}^{\Hsf} - y_{kk}}\right) \bone_k \bone_k^{\Hsf}
	+ \sum_{\ell \sim k} \v{\Psi}_{k\ell},
	\end{aligned}
}
\end{align}
and $\bone_k \in \Rset^n$ is the vector of all zeros, except the $k$-th entry that is unity. The notation $\ell \sim k$ indicates that a transmission line connects buses $\ell$ and $k$ in the power network. Here, $\Rset$ is the set of all real numbers.
Voltage magnitudes across the network are deemed to remain close to rated voltage levels as $\ul{v}_k \leq | V_k | \leq \ol{v}_k$ at bus $k$, that is equivalently written as
\begin{align}
\ul{v}_k^2 \leq \v{V}^{\Hsf} \bone_k \bone_k^{\Hsf} \v{V} \leq \ol{v}_k^2.
\end{align}
Consider two assets connected at each bus -- an uncontrollable asset whose apparent power draw is fixed and known and a controllable asset whose power injection can vary within known capacity limits. Let $p_k^D$ and $q_k^D$, respectively, denote the nominal real and reactive power draws at bus $k$ from the uncontrollable asset. Similarly, let $p_k^G$ and $q_k^G$ denote the real and reactive power generation at bus $k$, respectively, that vary within known capacity limits. Collecting these limits across all buses, we write $\left( \v{p}^G, \v{q}^G \right) \in \Gset$, where
\begin{align}
\Gset := \left\lbrace \left( \v{p}^G, \v{q}^G \right) \ \vert \ 
\ul{\v{p}} \leq \v{p}^G \leq \ol{\v{p}}, \ \ul{\v{q}} \leq \v{q}^G \leq \ol{\v{q}} \right\rbrace.
\label{eq:S.def}
\end{align} 
Associated with that generation is a  dispatch cost $c_k(p_k^G, q_k^G)$. Assume that $c_k$ is jointly convex in its arguments. Such costs in wholesale markets are inferred from supply offers and demand bids. Uncontrollable assets represent the collective inelastic power demands at a bus, while generators and proxy demand resources comprise controllable assets.  

The SO seeks to compute a dispatch that minimizes the aggregate dispatch costs from the collection of grid-connected controllable assets and meets the power requirements of the uncontrollable ones, meeting the  engineering constraints of the power network as follows.
\begin{subequations}
	\begin{alignat}{2}
	\vspace{-0.1in}
	\Pac  : \
	& {\text{minimize}} && \ \ \sum_{k=1}^n c_k(p^G_k, q^G_k), 
	\notag
	\\
	& \text{subject to} 
	&&\ \ \left( \v{p}^G, \v{q}^G \right) \in \Gset, 
	\label{eq:pqBounds}
	\\
	&&& \ \ p^G_k - p^D_k = \v{V}^{\Hsf} \v{\Phi}_k \v{V}, 
	\label{eq:pBalance}
	\\
	&&& \ \ q^G_k - q^D_k = \v{V}^{\Hsf} \v{\Psi}_k \v{V}, 
	\label{eq:qBalance}
	\\
	&&& \ \ \v{V}^{\Hsf} \v{\Phi}_{k\ell} \v{V} \leq f_{k \ell}, 			
	\label{eq:fLim}
	\\
	&&&\ \ \ul{v}_k^2 \leq \v{V}^{\Hsf} \bone_k \bone_k^{\Hsf} \v{V} \leq \ol{v}_k^2 
	\label{eq:v}
	\\			
	&&& \ \ \text{for } k = 1,\ldots,n, \ \ell \sim k
	\notag
	\end{alignat}%
	\label{eq:Pac}
\end{subequations}
over $\v{p}^G, \v{q}^G$ and $\v{V}$. The boldfaced symbols collect the corresponding variables across the network. 
$\Pac$ is nonconvex, owing to quadratic equalities. 
In what follows, we consider prices to support such a dispatch.

\section{Pricing Mechanisms}
\label{sec:prices}

We consider two candidate pricing mechanisms. The first set of prices are derived from Lagrange multipliers that satisfy the
Karush-Kuhn-Tucker (KKT) optimality conditions for $\Pac$ at one of its local optima. For an optimum, such multipliers exist under certain regularity conditions such as those in \cite[Proposition 4.3.13]{Bertsekas/99}.
The other set of prices are derived from an SDP-based convex relaxation of $\Pac$. We call these prices SDP-LMPs. 
AC-LMPs and SDP-LMPs are \emph{not always equal}. In this section, we characterize the relationship between these two pricing mechanisms.
All results in this paper are derived under the assumption that $\Pac$ admits a strictly feasible point.

\begin{figure*}[h]
	\centering
	\begin{framed}
		{\small
			\begin{itemize}
				\item Primal feasibility conditions: \eqref{eq:pBalance} -- \eqref{eq:v}.
				
				\item Dual feasibility: $\ol{\mu}_{k}^{p,\star}$, $\ul{\mu}_{k}^{p,\star}$, $\ol{\mu}_{k}^{q,\star}$, $\ul{\mu}_{k}^{q,\star}$, $\mu_{k \ell}^\star$, $\ol{\mu}^{v,\star}_k$, $\ul{\mu}^{v,\star}_k \geq 0$, for $k = 1, \ldots, n$, $k\ell=1, \ldots, 2m$,  
				
				\item Stationarity conditions: For $k = 1,...,n, (k\ell)=1,...,m$: 
				\begin{subequations}
					\begin{gather}
					\left[\sum_{k = 1}^n \Lambda_k^{p,\star} \v{\Phi}_{k}  + \sum_{k = 1}^n \Lambda_k^{q,\star} \v{\Psi}_{k} 
					+  \sum_{k \ell = 1}^m \mu_{k \ell}^\star \v{\Phi}_{k \ell}  + 
					\sum_{k=1}^n \ol{\mu}^{v,\star}_k \bone_k \bone_k^\mathsf{T} 
					- \sum_{k=1}^n \ul{\mu}^{v,\star}_k \bone_k \bone_k^\mathsf{T} \right] \v{V}^\star = 0,
					\label{eq:grad.2}
					\\
					\nabla_{{p_k^G}} \left[ {c_k(p_k^{G,\star}, q_k^{G,\star})} \right] - \Lambda_k^{p,\star}+\ol{\mu}_k^{p,\star}-\ul{\mu}_k^{p,\star}
					=
					\nabla_{{q_k^G}} \left[ {c_k(p_k^{G,\star}, q_k^{G,\star})} \right] -\Lambda_k^{q,\star}+\ol{\mu}_k^{q,\star}-\ul{\mu}_k^{q,\star}=0.
					\label{eq:grad.3}
					\end{gather}
				\end{subequations}
				
				\item Complementary slackness conditions: For $k = 1,...,n, (k\ell)=1,...,m$:
				\begin{subequations}
					\begin{gather}
					\mu_{k\ell}^\star[ \v{V}^{\Hsf, \star}\v{\Phi}_{k \ell} \v{V}^\star - f_{k\ell}] 
					= 
					\ol{\mu}^v_k \left( \v{V}^{\Hsf, \star} \bone_k \bone_k^{\Hsf} \v{V}^\star - \ol{v}_k^{2} \right)
					=
					\ul{\mu}^v_k \left(\v{V}^{\Hsf, \star} \bone_k \bone_k^{\Hsf} \v{V}^\star - \ul{v}_k^{2} \right) 
					= 0,
					\label{eq:CS.1}
					\\
					\ul{\mu}^{q,\star}_k \left(q^{G,\star}_k - \ul{q}^G_k \right) 
					= 
					\ol{\mu}^{q,\star}_k \left(q^{G,\star}_k - \ol{q}^G_k \right) 
					= 
					\ul{\mu}^{p,\star}_k \left(p^{G,\star}_k - \ul{p}^G_k \right) 
					= 
					\ol{\mu}^{p,\star}_k \left(p^{G,\star}_k - \ol{p}^G_k \right) = 0.
					\label{eq:CS.2}
					\end{gather}
				\end{subequations}
			\end{itemize}
		}
	\end{framed}
	\caption{The KKT conditions for $\Pac$.}
	\label{fig:KKT.AC}
\end{figure*}

\subsection{Locational Marginal Prices From nonconvex $\Pac$}
Associate Lagrange multipliers $\ol{\mu}^p_k$, $\ul{\mu}^p_k$, $\ol{\mu}^q_k$, $\ul{\mu}^q_k$ to the upper and lower, real and reactive capacity limits in \eqref{eq:pqBounds}, $\Lambda^p_k$, $\Lambda^q_k$ with \eqref{eq:pBalance}, \eqref{eq:qBalance},  $\mu_{k\ell}$ with \eqref{eq:fLim}, and $\ol{\mu}_k^v$, $\ul{\mu}_k^v$ with the upper and lower voltage limits in \eqref{eq:v}, respectively. The KKT conditions for $\Pac$ are given in Figure \ref{fig:KKT.AC}.
\begin{definition}[AC-LMPs] The Lagrange multipliers $\v{\Lambda}^{p,\star}$ and $\v{\Lambda}^{q,\star}$, that satisfy the KKT conditions for $\Pac$ for a locally optimal dispatch $\v{p}^{G, \star}, \v{q}^{G, \star}, \v{V}^\star$ define the AC locational marginal prices (AC-LMPs) for real and reactive power, respectively, for that dispatch, assuming these multipliers exist.
\end{definition}
By definition, these prices are associated with specific local minima of $\Pac$. Denote by $J^\star_{\AC} (\v{p}^D, \v{q}^D)$, the  cost of \eqref{eq:Pac} at a \emph{global} minimum of $\Pac$, parameterized by the nodal real and reactive power demands. The feasible set of $\Pac$ is compact. Assuming that this set varies continuously in nodal demands, $J^\star_{\AC} (\v{p}^D, \v{q}^D)$ must then vary continuously with $\v{p}^D, \v{q}^D$, per \cite[Chapter 5]{still2018lectures}. 
Since $\Pac$ is nonconvex, the parametric optimal function $J^\star_\AC$ can be nonconvex. In general, it is also non-smooth. Under regularity conditions (see \cite[Proposition 3.3.3]{Bertsekas/99}), AC-LMPs associated with a global minimum are the marginal sensitivities of this optimal cost to nodal power demands, i.e., 
$\v{\Lambda}^{p,\star} = \nabla_{\v{p}^D} J^\star_{\AC}(\v{p}^D, \v{q}^D)$ and $\v{\Lambda}^{q,\star} = \nabla_{\v{q}^D} J^\star_{\AC} (\v{p}^D, \v{q}^D)$,
if $J^\star_{\AC}$ is differentiable. Here, $\nabla$ computes the gradient of its argument. 

\subsection{SDP Relaxation-Based Locational Marginal Prices}
We now define nodal prices for real and reactive powers from an SDP-based convex relaxation $\Psdp$ of $\Pac$  in \eqref{eq:Psdp}. To arrive at the relaxation, write
$ \v{V}^{\Hsf} \v{M} \v{V}$ as  $\trace( \v{M} \v{V} \v{V}^{\Hsf}) =  \trace( \v{M} \v{W}) $
for any $\v{M} \in \Cset^{n \times n}$ and $\v{W} = \v{V}\v{V}^{\Hsf}$, where $\trace$ computes the trace of a matrix. The above representation reduces quadratic forms in $\v{V}$ to linear forms in $\v{W} \in \Hset^n $ that is positive semidefinite  (henceforth denoted as $\v{W} \succeq 0$) and rank-1. Thus, $\Pac$ can be reformulated as a rank constrained SDP in $\v{W}$. Dropping the rank constraint gives
\begin{subequations}
	\begin{alignat}{2}
	\hspace{-0.09in}
	\Psdp  : \
	& {\text{minimize}} && \ \ \sum_{k=1}^n c_k(p^G_k, q^G_k), 
	\notag
	\\
	& \text{subject to} 
	&& \ \ \left( \v{p}^G, \v{q}^G \right) \in \Gset, 
	\label{eq:W.pqBounds}
	\\
	&&& \ \ p^G_k - p^D_k = \trace(\v{\Phi}_k \v{W}),  
	\label{eq:W.pBalance}
	\\
	&&& \ \ q^G_k - q^D_k = \trace( \v{\Psi}_k \v{W}),
	\label{eq:W.qBalance}
	\\
	&&& \ \ \trace(\v{\Phi}_{k\ell} \v{W} ) \leq f_{k \ell}, 			
	\label{eq:W.fLim}
	\\
	&&& \ \ \ul{v}_k^2 \leq \trace( \bone_k \bone_k^{\Hsf} \v{W}) \leq \ol{v}_k^2,  
	\label{eq:W.v}
	\\			
	&&& \ \ \v{W} \succeq 0,
	\label{eq:W.psd}
	\\
	&&& \ \ \text{for } k = 1,\ldots,n, \ \ell \sim k 
	\notag
	\end{alignat}
	\label{eq:Psdp}
\end{subequations}
over $\v{W}, \v{p}^G, \v{q}^G$. In contrast to $\Pac$, the optimization problem $\Psdp$ is convex. Associate the same Lagrange multipliers as for $\Pac$, but use $\v{\lambda}^p, \v{\lambda}^q$ instead of $\v{\Lambda}^p, \v{\Lambda}^q$ for the real and reactive power balance constraints \eqref{eq:W.pBalance} and \eqref{eq:W.qBalance} in $\Psdp$. In addition, associate $\v{U} \in \Hset^n$ as the matrix multiplier for the constraint $\v{W} \succeq 0$. The KKT optimality conditions for $\Psdp$ are then given by that for $\Pac$ in Figure \ref{fig:KKT.AC}, but with the following changes: \emph{(i)} $\v{\lambda}$'s replace $\v{\Lambda}$'s, \emph{(ii)} \eqref{eq:grad.2} changes to
\begin{align}
{
	\begin{aligned}
	\sum_{k = 1}^n \lambda_k^{p,\star} \v{\Phi}_{k} 
	+ \sum_{k = 1}^n \lambda_k^{q,\star} \v{\Psi}_{k} 
	+  \sum_{k \ell = 1}^m \mu_{k \ell}^\star \v{\Phi}_{k \ell}  
	+ \sum_{k=1}^n \ol{\mu}^{v,\star}_k \bone_k \bone_k^\mathsf{T} 
	- \sum_{k=1}^n \ul{\mu}^{v,\star}_k \bone_k \bone_k^\mathsf{T} - \v{U}^\star = 0,
	\end{aligned}
}
\end{align}
\emph{(iii)} $\trace(\v{M} \v{W}^\star )$ replaces $\v{V}^{\Hsf, \star} \v{M} \v{V}^\star$ in \eqref{eq:CS.1} for each quadratic form in $\v{V}$ and \emph{(iv)} the dual feasibility constraint $\v{U} \succeq 0$ is added to the list. 
\begin{definition}[SDP-LMPs] The Lagrange multipliers $\v{\lambda}^{p,\star}$ and $\v{\lambda}^{q,\star}$, that satisfy the KKT conditions for $\Psdp$ define the SDP locational marginal prices (SDP-LMPs) for real and reactive power, respectively.
\end{definition}
Unlike AC-LMPs, the SDP-LMPs are \emph{not} associated with a local minimum of $\Pac$.
As a result, these prices do not change with the local optimal dispatch that a nonlinear optimization solver may find. Instead, they are purely functions of the problem parameters and are robust to convergence properties of the optimization solver for $\Pac$.

Let $J^\star_{\SDP}(\v{p}^D,\v{q}^D)$ denote the optimal cost of $\Psdp$. The nature of the constraints of $\Psdp$ guarantee that $J^\star_\SDP$ is jointly convex in its arguments. It can, however, be non-smooth. SDP-LMPs are the marginal sensitivities of the optimal cost of the SDP relaxation to nodal real and reactive powers as a result of the envelope theorem (see \cite[Chapter 7]{still2018lectures}),  i.e.,
$\v{\lambda}^{p,\star} = \nabla_{\v{p}^D} J^\star_{\SDP}(\v{p}^D,\v{q}^D )$ and $\v{\lambda}^{q,\star} = \nabla_{\v{q}^D} J^\star_{\SDP} (\v{p}^D,\v{q}^D)$,
if $J^\star_{\SDP}$ is differentiable at $(\v{p}^D, \v{q}^D)$. 

In what follows, we analyze the relationship between AC-LMPs and SDP-LMPs. Note that we study pricing schemes in this paper that associate prices for both real and reactive powers. In part, our choice is motivated to analyze a generalization of SOCP-based DLMPs in \cite{PapavasiliouDLMPs} that does the same. We refer the reader to celebrated debates on reactive power pricing in \cite{Mount,Zhong,Lipka}. Here, we sidestep such debates and focus on the mathematical properties of these prices.


\begin{figure*}
	\begin{align}
	{\small
		\begin{aligned}
		&\Lcal_V(\v{p}^G, \v{q}^G, \v{V}, \v{\Lambda}^p, \v{\Lambda}^q, \v{\mu}, \ol{\v{\mu}}^v, \ul{\v{\mu}}^v) 
		\\
		&:= \sum_{k=1}^{n} c_k(p^G_k, q^G_k) -\sum_{k=1}^{n} \Lambda_k^p \left(p^G_k - p^D_k - \v{V}^{\Hsf}\v{\Phi}_k \v{V}\right)
		- \sum_{k=1}^{n} \Lambda_k^q \left(q^G_k - q^D_k - \v{V}^{\Hsf}\v{\Psi}_k \v{V}\right) 
		\\
		& + \sum_{k\ell=1}^{m} \mu_{k\ell} \left(\v{V}^{\Hsf}\v{\Phi}_{k\ell} \v{V}\ - f_{k\ell} \right)  
		+ \sum_{k=1}^{n} \ol{\mu}^{v}_k \left(\v{V}^{\Hsf} \bone_k \bone_k^{\Hsf} \v{V}- \ol{v}_k^2 \right) 
		- \sum_{k=1}^{n} \ul{\mu}^{v}_k \left(\v{V}^{\Hsf} \bone_k \bone_k^{\Hsf} \v{V} - \ul{v}_k^2 \right).
		\end{aligned}
	}
	\label{eq:partialLang.V}
	\end{align}
	\begin{align}
	{\small
		\begin{aligned}
		&\Lcal_W(\v{p}^G, \v{q}^G, \v{W}, \v{\lambda}^p, \v{\lambda}^q, \v{\mu}, \ol{\v{\mu}}^v, \ul{\v{\mu}}^v, \v{U}) 
		\\
		&:= \sum_{k=1}^{n} c_k(p^G_k, q^G_k) 
		- \sum_{k=1}^{n} \lambda_k^p \left(p^G_k - p^D_k - \trace(\v{\Phi}_k \v{W})\right)  
		- \sum_{k=1}^{n} \lambda_k^q \left(q^G_k - q^D_k - \trace(\v{\Psi}_k \v{W})\right) 
		\\
		& 
		+ \sum_{k=1}^{n} \mu_{k\ell} \left( \trace (\v{\Phi}_{k\ell} \v{W}) - f_{k\ell} \right)  
		+ \sum_{k=1}^{n} \ol{\mu}^{v}_k \left(\trace(\bone_k \bone_k^{\Hsf} \v{W}) - \ol{v}_k^2 \right)  
		- \sum_{k=1}^{n} \ul{\mu}^{v}_k \left(\trace(\bone_k \bone_k^{\Hsf} \v{W}) - \ul{v}_k^2 \right)
		- \trace(\v{UW}). 
		\end{aligned}
	}
	\label{eq:partialLang.W}
	\end{align}
	\hrule
\end{figure*}

\subsection{Relationship Between AC-LMPs and SDP-LMPs}
Define the partial Lagrangian functions for $\Pac$ and $\Psdp$ in \eqref{eq:partialLang.V} and \eqref{eq:partialLang.W}, respectively. Using $\Lcal_V$, $\Pac$ admits the standard min-max reformulation as
\begin{align}
\Pac : \underset{\substack{\v{V} \in \Cset^n, \\ \left(\v{p}^G, \v{q}^G\right)  \in \Gset}}{\inf} \ \ \underset{\substack{\v{\mu},  \ol{\v{\mu}}^v, \ul{\v{\mu}}^v \geq 0 \\
		\v{\lambda}^p, \v{\lambda}^q}}{\sup} \ \ \Lcal_V.
\label{Pac.minmax}
\end{align}
Then, the dual program of $\Pac$ is 
\begin{align}
\Dcal\Pac : \underset{\substack{\v{\mu},  \ol{\v{\mu}}^v, \ul{\v{\mu}}^v \geq 0 \\
		\v{\Lambda}^p, \v{\Lambda}^q}}{\sup} \ \
\underset{\substack{\v{V} \in \Cset^n, \\ \left(\v{p}^G, \v{q}^G\right)  \in \Gset}}{\inf} \ \ \Lcal_V.  
\label{eq:DPac}
\end{align}
Similarly, $\Psdp$ and its dual are given by 
\begin{align}
\Psdp &: \underset{\substack{\v{W} \in \Hset^n, \\ \left(\v{p}^G, \v{q}^G\right)  \in \Gset}}{\inf} \ \ \underset{\substack{\v{\mu},  \ol{\v{\mu}}^v, \ul{\v{\mu}}^v \geq 0 \\
		\v{U} \succeq 0, \v{\lambda}^p, \v{\lambda}^q}}{\sup} \ \ \Lcal_W,
\\
\Dcal\Psdp &: \underset{\substack{\v{\mu},  \ol{\v{\mu}}^v, \ul{\v{\mu}}^v \geq 0 \\
		\v{U} \succeq 0, \v{\lambda}^p, \v{\lambda}^q}}{\sup} \ \
\underset{\substack{\v{W} \in \Hset^n, \\ \left(\v{p}^G, \v{q}^G\right)  \in \Gset}}{\inf} \ \ \Lcal_W. 
\end{align}
Having defined these primal and dual problems, we now establish relationships between AC-LMPs and SDP-LMPs. Our exposition makes use of the following notations. For an arbitrary extended real-valued function $h:\Rset^r \to \Rset \cup \{ \pm \infty \}$, its epigraph is given by 
\begin{equation}
\epi h := \left\lbrace (\v{x},t)| \v{x} \in \dom h \subseteq \Rset^r, h(\v{x}) \leq t \right\rbrace.
\end{equation}
Here, $\dom h$ is the domain of $h$, over which $h$ assumes finite values. 
Also, for an arbitrary set  $\Mset$, let $\conv \Mset$ denote its convex hull--the smallest convex set that contains $\Mset$.

\begin{theorem}
	\label{thm:epigraph}
	The following assertions hold:
	\begin{enumerate}[label=(\alph*)]
		\item $\Dcal\Pac$ and $\Dcal\Psdp$ are equivalent optimization problems.
		\item $\conv \epi {J^\star_{\AC}(\v{p}^D, \v{q}^D)}  \subseteq \epi {J^\star_{\SDP}(\v{p}^D, \v{q}^D)}$.
		\item  When $\Pac$ has zero duality gap, i.e., $\Psdp$ admits a solution with $\rank \v{W}^\star = 1$, then SDP-LMPs are also AC-LMPs associated with a global optimum of $\Pac$.

	\end{enumerate} 
\end{theorem}
\begin{proof} For part (a), we write $\Lcal_V$ as $\v{V}^{\Hsf} \widehat{\v{U}} \v{V} + \zeta$, where
	\begin{align} 
	{\small
		\begin{aligned}
		&\widehat{\v{U}}
		:= \sum_{k=1}^{n} \left[ \Lambda_k^p \v{\Phi}_k 
		+ \Lambda_k^q \v{\Psi}_k 
		+ \left(\ol{\mu}^{v}_k - \ul{\mu}^{v}_k \right)\bone_k \bone_k^{\Hsf} \right]
		+ \sum_{k\ell=1}^{m} \mu_{k\ell} \v{\Phi}_{k\ell},
		\end{aligned}
		\label{eq:Uhat.def}
	}
	\end{align}
	\begin{align}
	{\small
		\begin{aligned}
		\zeta(\v{p}^G, \v{q}^G,  \v{\Lambda}^p, \v{\Lambda}^q, \v{\mu}, \ol{\v{\mu}}^v, \ul{\v{\mu}}^v) 
		& := \sum_{k=1}^{n} \left[c_k(p^G_k, q^G_k) -  \Lambda_k^p \left( p^G_k - p^D_k \right) - \Lambda_k^q \left( q^G_k - q^D_k \right) \right]
		\\ 
		& \qquad  - \sum_{k=1}^{n} \left(\ol{\mu}^{v}_k  \ol{v}_k^2 - \ul{\mu}^{v}_k \ul{v}_k^2\right)
		- \sum_{k\ell=1}^{m} \mu_{k\ell}  f_{k\ell}.
		\end{aligned}
	}
	\label{eq:zeta.def}
	\end{align}
	Then, we have
	\begin{align}
	\underset{\v{V} \in \Cset^n}{\inf} \  \Lcal_V 
	=
	\begin{cases}
	\zeta, & \text{if } \widehat{\v{U}} \succeq 0,
	\\
	- \infty, & \text{otherwise}.
	\end{cases}
	\end{align}
	and $\Dcal \Pac$ becomes 
	\begin{alignat}{2}
	{\small
		\begin{aligned}
		\underset{\substack{\v{\mu},  \ol{\v{\mu}}^v, \ul{\v{\mu}}^v \geq 0 \\
				\widehat{\v{U}}, \v{\Lambda}^p, \v{\Lambda}^q}}{\sup}
		\left\{ \underset{\substack{\left(\v{p}^G, \v{q}^G\right)  \in \Gset}}{\inf}  \ \
		\zeta \right\}, 
		\
		\text{subject to} \ \ \eqref{eq:Uhat.def}, \ 
		\widehat{\v{U}} \succeq 0.
		\end{aligned}
	}
	\label{eq:DPac.2}
	\end{alignat}
	To show the equivalence of the above problem with $\Dcal\Psdp$, note that $\Lcal_W$ is linear in $\v{W}$ and thus, unconstrained minimization of $\Lcal_W$ over $\v{W} \in \Hset^n$ yields $-\infty$, unless $\nabla_\v{W} \Lcal_W = 0$. Setting that derivative to zero, we recover \eqref{eq:Uhat.def} with $\v{U}$ instead of $\widehat{\v{U}}$ and $\v{\Lambda}$'s replaced by $\v{\lambda}$'s. Incorporating this as a constraint in $\Dcal\Psdp$ yields \eqref{eq:DPac.2} with  $\v{U}$ instead of $\widehat{\v{U}}$ and $\v{\Lambda}$'s replaced by $\v{\lambda}$'s. This completes the proof of the dual equivalence of $\Pac$ and $\Psdp$.

	For part (b), we appeal to weak duality and conclude that $J_\AC^\star$, the optimal value of $\Pac$, dominates the optimal value of $\Dcal\Pac$. From part (a), the latter equals the optimal value of $\Dcal\Psdp$, which equals $J_\SDP^\star$, because strong duality holds for $\Psdp$. Strong duality follows from Slater's condition (see \cite[Theorem 2.165]{bonnans2013perturbation}) that applies under our hypothesis that $\Pac$ admits a strictly feasible point. Thus, we have
	\begin{align}
	J_\SDP^\star(\v{p}^D, \v{q}^D) \leq J_\AC^\star(\v{p}^D, \v{q}^D).
	\label{eq:Jsdp.less.Jac}
	\end{align}
	Since $J^\star_\SDP$ is convex in its arguments, its epigraph is a convex set. Thus, \eqref{eq:Jsdp.less.Jac} implies $ \epi{J_\AC^\star} \subseteq \epi{J_\SDP^\star} $. The rest follows from the fact that if a convex set contains a nonconvex set, then the former contains the convex hull of the latter. Part (c) is a consequence of the fact that if $\rank \v{W}^\star = 1$, then its spectral decomposition  $\v{W}^\star = \v{V}^\star \v{V}^{\Hsf, \star}$ yields $\v{V}^\star$ that together with $\v{\lambda}^\star$, $\v{\mu}^\star$, $\v{\ul{\mu}}^{v,\star}$ and $\v{\ol{\mu}}^{v,\star}$ satisfies the KKT system for $\Pac$.
\end{proof}

We now contextualize Theorem \ref{thm:epigraph} within existing literature. The dual equivalence in part (a) between $\Pac$ and $\Psdp$ has been reported before, e.g., see \cite{lavaei2012zero, wolkowicz2012handbook}. We include it for completeness and now contrast its implications  with a similar result known for convex hull pricing (CHP). In \cite{Gribik,Schiro}, CHP tackles the nonconvexity introduced by integer unit commitment decisions with linearized power flow equations. CHPs are derived from the convex Lagrangian dual problem of the unit commitment problem--a property that part (a) suggests for our context, where  SDP-LMPs are derived from a problem equivalent to the Lagrangian dual problem of $\Pac$. In this respect, SDP-LMPs and CHPs are indeed similar.

The aforementioned similarity between SDP-LMPs and CHPs might suggest that SDP-LMPs will inherit other properties of CHP. For example, the analysis in \cite{Gribik} might indicate that the convex hull of the epigraph of $J^\star_{\AC}(\v{p}^D, \v{q}^D)$  would equal the epigraph of $J^\star_{\SDP}(\v{p}^D, \v{q}^D)$. In fact, CHP derives its name from the relation between the convex hull of the epigraph of the nonconvex problem and its convex Lagrangian dual. However, part (b) only proves an inclusion instead of an equality. We now argue why the analysis in \cite{Gribik} does \emph{not} carry over to our setting.

For an extended real-valued function $h:\Rset^r \to \Rset \cup \{ \pm \infty \}$, define its \emph{Fenchel conjugate} as 
\begin{align}
h^{\fc}(\v{\xi}) := \underset{\v{x} \in \dom h}{\sup} \left\{ \v{\xi}^\T \v{x} - h(\v{x}) \right\} \in \Rset \cup \{ \pm \infty \}.
\end{align}
Extending the definition, one can also define the Fenchel \emph{biconjugate} of $h$ as $h^{\fc\fc}$. Per the Fenchel-Moreau-Rockafellar theorem in \cite[Theorem 2.113]{bonnans2013perturbation}, we have
\begin{align}
\conv \epi h  = \epi h^{\fc\fc}
\label{eq:epi.conv.cc}
\end{align}
for a continuous function $h$. This characterization of the convex hull of the epigraph of a function proves useful to analyze the epigraphs of $J^\star_\AC$ and $J^\star_\SDP$. Specifically, consider the global optimal cost of $\Pac$, parameterized as $J_\AC^\star(\v{p}^D, \v{q}^D, \v{f}, -\ul{\v{v}}^2, \ol{\v{v}}^2)$. Define the same for $J_\SDP^\star$. 
\begin{proposition} 
	\label{lemma:biconj}
	The parametric optimal costs satisfy 
	\begin{align}
	\hspace{-0.1in}
	{\small
		J_{\AC}^{\star,\fc\fc} (\v{p}^D, \v{q}^D, \v{f}, -\ul{\v{v}}^2, \ol{\v{v}}^2) = J_{\SDP}^{\star,\fc\fc} (\v{p}^D, \v{q}^D, \v{f}, -\ul{\v{v}}^2, \ol{\v{v}}^2).
	}
	\end{align}
\end{proposition}
\begin{proof}
	Define the support function of set $\Aset$ as
	\begin{align}
	\delta_\Aset(a) := 
	\begin{cases}
	0, & \text{if } a \in \Aset, \\
	+\infty, & \text{otherwise}.
	\end{cases}
	\label{eq:supp.f}
	\end{align}
	Using this notation, $\Pac$ becomes
	\begin{align}
	{\small
		\begin{aligned}
		J^\star_\AC(\v{p}^D, \v{q}^D, \v{f}, -\ul{\v{v}}^2, \ol{\v{v}}^2) 
		& = \underset{ \substack{ \left( \v{p}^G, \v{q}^G \right) \in \Gset \\ \v{V} \in \Cset^n }}{ \inf } 
		\left \lbrace \sum_{k=1}^n c_k(p_k^G, q_k^G) + \sum_{k=1}^n  \delta_{\{0\}} \left( p_k^D - p_k^G + \v{V}^{\Hsf} \Phi_k \v{V} \right) \right.  
		\\
		& \qquad + \sum_{k=1}^n \delta_{\{0\}} \left(q_k^D - q_k^G + \v{V}^{\Hsf} \Psi_k \v{V} \right) + \sum_{k\ell=1}^{m} \delta_{\Rset_+} \left(f_{k\ell} - \v{V}^{\Hsf}\v{\Phi}_{k\ell} \v{V}\right)  
		\\
		& \qquad \left. + \sum_{k=1}^{n} \delta_{\Rset_+} \left(\ol{v}_k^2 - \v{V}^{\Hsf} \bone_k \bone_k^{\Hsf} \v{V} \right) 
		+ \sum_{k=1}^{n} \delta_{\Rset_+} \left(\v{V}^{\Hsf} \bone_k \bone_k^{\Hsf} \v{V} - \ul{v}_k^2 \right)\right\rbrace.
		\end{aligned}
	}
	\end{align}
	Applying Lemma \ref{lemma:param.fenchel} in the appendix, together with the relations $\delta_{\{0\}}^{\fc}(a) = 0$ and $\delta_{\Rset_+}^{\fc}(a) = \delta_{\Rset_+}(-a)$, we get
	\begin{align}
	\hspace{-0.07in}
	{\small
		\begin{aligned}
		J^{\star,\fc}_\AC(\v{\Lambda}^p, \v{\Lambda}^q, -\v{\mu}, -\v{\ul{\mu}}^v, -\v{\ol{\mu}}^v)
		&= \delta_{\Rset_+}(\mu_{k \ell})
		+ \delta_{\Rset_+}(\ol{\mu}^v_k)  
		+ \delta_{\Rset_+}(\ul{\mu}^v_k)
		\\
		& \quad - \underset{ \substack{ \left( \v{p}^G, \v{q}^G \right) \in \Gset \\ \v{V} \in \Cset^n }}{ \textrm{inf} } 
		\left \lbrace \sum_{k=1}^n c_k(p_k^G, q_k^G) + \sum_{k=1}^n  \Lambda^p_k \left( - p_k^G + \v{V}^{\Hsf} \Phi_k \v{V} \right) \right. 
		\\
		& \qquad \qquad + \sum_{k=1}^n \Lambda^q_k \left( - q_k^G + \v{V}^{\Hsf} \Psi_k \v{V} \right) + \sum_{k\ell=1}^{m} \mu_{k \ell} \left( \v{V}^{\Hsf}\v{\Phi}_{k\ell} \v{V}\right) 
		\\
		& \qquad \qquad \left. + \sum_{k=1}^{n} \left(\ol{\mu}^v_k - \ul{\mu}^v_k\right) \v{V}^{\Hsf} \bone_k \bone_k^{\Hsf} \v{V}
		\right\rbrace.
		\end{aligned}
	}
	\end{align}
	Using the definition of $\Lcal_V$ in \eqref{eq:partialLang.V}, the above equation yields
	\begin{align}
	\begin{aligned}
	J^{\star,\fc}_\AC(\v{\Lambda}^p, \v{\Lambda}^q, -\v{\mu}, -\v{\ul{\mu}}^v, -\v{\ol{\mu}}^v)
	&= 	\delta_{\Rset_+}(\mu_{k \ell})
	+ \delta_{\Rset_+}(\ol{\mu}^v_k)  
	+ \delta_{\Rset_+}(\ul{\mu}^v_k)
	\\
	& \quad + \sum_{k=1}^n \left[ {\Lambda}^p_k \v{p}^D_k + {\Lambda}^q_k \v{q}^D_k  - \ol{\mu}_k^v \ol{v}_k^2 + \ul{\mu}_k^v \ul{v}_k^2 \right]
	- \sum_{k \ell = 1}^m \mu_{k \ell} f_{k\ell} 
	\\
	& \quad - \underset{ \substack{ \left( \v{p}^G, \v{q}^G \right) \in \Gset \\ \v{V} \in \Cset^n }}{ \inf }  \ \ \Lcal_V(\v{p}^G, \v{q}^G, \v{V}, \v{\Lambda}^p, \v{\Lambda}^q, \v{\mu}, \ol{\v{\mu}}^v, \ul{\v{\mu}}^v).
	\end{aligned}
	\end{align}
	Thus, its biconjugate is given by
	\begin{align}
	{\small
		\begin{aligned}
		J^{\star,\fc \fc}_\AC(\v{p}^D, \v{q}^D, \v{f}, -\ul{\v{v}}^2, \ol{\v{v}}^2)
		= \underset{\substack{\v{\mu},  \ol{\v{\mu}}^v, \ul{\v{\mu}}^v \geq 0 \\
				\v{\Lambda}^p, \v{\Lambda}^q}}{\sup} \ \
		\underset{\substack{\v{V} \in \Cset^n, \\ \left(\v{p}^G, \v{q}^G\right)  \in \Gset}}{\inf} \ \ \Lcal_V.
		\end{aligned}
	}
	\end{align}
	The RHS of the above equation is the optimal cost of $\Dcal\Pcal_{\AC}$ in \eqref{eq:DPac}. By virtue of Theorem \ref{thm:epigraph}(a), this cost coincides with the optimal cost of $\Dcal\Pcal_\SDP$. Strong duality of $\Psdp$ then gives
	\begin{align}
	\hspace{-0.1in}
	{\small
		J^{\star,\fc \fc}_\AC(\v{p}^D, \v{q}^D, \v{f}, -\ul{\v{v}}^2, \ol{\v{v}}^2) = J^{\star}_\SDP(\v{p}^D, \v{q}^D, \v{f}, -\ul{\v{v}}^2, \ol{\v{v}}^2). 
	}
	\end{align}
	Recall that $J^{\star}_\SDP$ is convex and continuous. Hence, $J^{\star}_\SDP = J^{\star, \fc\fc}_\SDP$, per \cite[Theorem 2.113]{bonnans2013perturbation}, completing the proof.
\end{proof}
The parametric optimal dual cost is known to provide the tightest convex lower bound on the parametric optimal primal cost of a nonconvex program that is linearly parameterized on the right-hand side. Thus, Proposition \ref{lemma:biconj} is not surprising in light of Theorem \ref{thm:epigraph}(a) that establishes the equivalence between the dual problems of $\Pac$ and $\Psdp$.  
Combining this result with \eqref{eq:epi.conv.cc}, the convex hull of the epigraph of $J_\AC^\star$ indeed equals the epigraph of $J_\SDP^\star$, but only when viewed as a function of \emph{all} parameters listed in Proposition \ref{lemma:biconj}, i.e.,
\begin{align}
\begin{aligned}
&\conv \epi J_{\AC}^{\star,\fc\fc} (\v{p}^D, \v{q}^D, \v{f}, -\ul{\v{v}}^2, \ol{\v{v}}^2) 
\\
& \quad = \epi J_{\SDP}^{\star,\fc\fc} (\v{p}^D, \v{q}^D, \v{f}, -\ul{\v{v}}^2, \ol{\v{v}}^2).
\end{aligned}
\end{align}
Fixing a subset of these parameters 
amounts to taking a \emph{slice} of these sets. Convex hull of the slice of a nonconvex set may \emph{not} always equal the slice of its convex hull (see Figure \ref{fig:slice})--a relation required to claim equality between  $\conv \epi J_{\AC}^{\star,\fc\fc} (\v{p}^D, \v{q}^D)$ and $\epi J_{\SDP}^{\star,\fc\fc} (\v{p}^D, \v{q}^D)$. 
\begin{figure}[ht]
	\centering
	\includegraphics[width=0.55\textwidth]{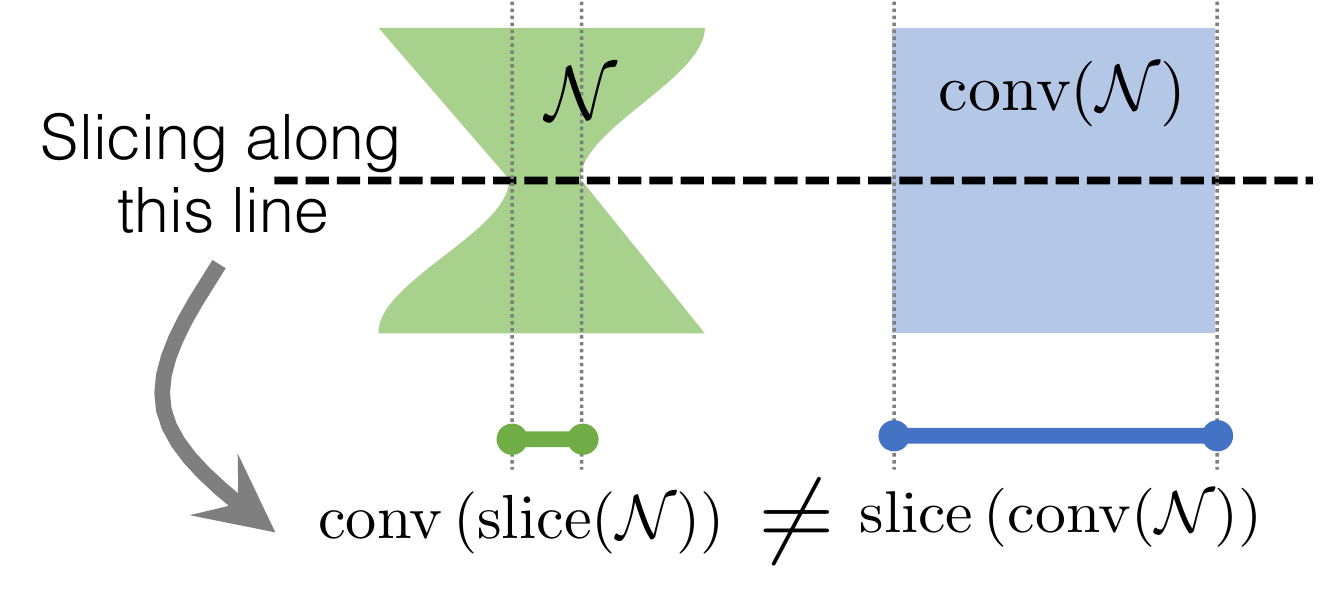}
	\caption{Figure illustrating that slices may not preserve the relation between a nonconvex set and its convex hull.}
	\label{fig:slice}
\end{figure}


\section{Revenue Adequacy of the Pricing Mechanisms}
\label{sec:ra}
In this section and the next, we study properties of these prices that are relevant to electricity market operations. We say a market mechanism is \emph{revenue adequate}, if the rents collected from power sales are enough to cover the rents payable to suppliers. Revenue adequacy ensures that the SO never runs cash negative after settling the payments of market participants.
To present our results on revenue adequacy, we first define the settlements of market participants with AC-LMPs and SDP-LMPs. 

Consider a local optimal dispatch of $\Pac$, given by $(p_k^{G,\star}, q_k^{G,\star}, \v{V}^\star )$. With AC-LMPs $\v{\Lambda}^{p,\star}$, $\v{\Lambda}^{q,\star}$ associated with that  dispatch, the controllable asset at bus $k$ is paid
\begin{equation}
\pi_k^G :=  \Lambda_{k}^{p,\star} p_k^{G,\star} + \Lambda_{k}^{q,\star}  q_k^{G,\star} \label{eq:pi.cont}
\end{equation}
by the SO. Similarly, uncontrollable asset with its demand $p_k^D$ and $q_k^D$ pays to the SO,
\begin{equation}\pi_k^D :=  \Lambda_{k}^{p,\star}  p_{k}^{D} + \Lambda_{k}^{q,\star}  q_k^{D}. 
\label{eq:pi.uncont}
\end{equation}
For payments based on SDP-LMPs, replace $\Lambda_{k}^{p,\star}, \Lambda_{k}^{q,\star}$ in \eqref{eq:pi.cont}-\eqref{eq:pi.uncont}  with $\lambda_{k}^{p,\star}, \lambda_{k}^{q,\star}$, respectively. Note that with SDP-LMPs, we consider payments defined using prices obtained from the SDP relaxation, but use these prices together with a locally optimal dispatch of $\Pac$ to calculate payments. 
With these payments, a pricing mechanism is revenue adequate if the merchandising surplus (MS) is nonnegative, i.e., if
\begin{equation}
\textrm{MS} := \sum_{k=1}^{n} \left(\pi_k^G - \pi_k^D \right) \geq 0.
\label{eq:MS.def}
\end{equation}

\begin{theorem}
	\label{thm:2}
	If voltage lower limits are non-binding at all buses at a local optimal solution of $\Pac$, i.e., 
	$|V_k^\star| > \ul{v}_k$
	for $k=1,\ldots,n$, then AC-LMPs define a revenue adequate pricing mechanism. 
\end{theorem}

\begin{proof}
	Expanding MS in \eqref{eq:MS.def}, we get
	\begin{subequations}
		\begin{align}
		\textrm{MS} 
		&= \sum_{k=1}^{n} \Lambda^{p,\star}_k \left( p_k^D - p_k^{G,\star}\right) + \sum_{k=1}^{n} \Lambda^{q,\star}_k \left( q_k^D - q_k^{G,\star}\right) 
		\notag
		\\
		&= - \sum_{k=1}^{n}  \Lambda^{p,\star}_k \v{V}^{\Hsf,\star} \v{\Phi}_k \v{V}^\star 
		- \Lambda^{q,\star}_k \v{V}^{\Hsf,\star} \v{\Psi}_k \v{V}^\star 
		\label{eq:MS.1}
		\\
		&=  \sum_{k \ell = 1}^m \mu_{k \ell}^\star  \v{V}^{\Hsf, \star} \v{\Phi}_{k \ell} \v{V}^\star 
		+  \sum_{k=1}^n  \ol{\mu}^{v,\star}_k \v{V}^{\Hsf,\star}  \bone_k \bone_k^\mathsf{T} \v{V^\star} 
		-  \sum_{k=1}^n  \ul{\mu}^{v,\star}_k   \v{V}^{\Hsf,\star} \bone_k \bone_k^\mathsf{T} \v{V^\star}
		\label{eq:MS.2}
		\\
		& = \sum_{k \ell = 1}^m \mu_{k \ell}^\star f_{k\ell}
		+ \sum_{k=1}^n \ol{\mu}^{v,\star}_k \ol{v}_k^2 - \sum_{k=1}^n \ul{\mu}^{v,\star}_k \ul{v}_k^2.
		\label{eq:MS.3}
		\end{align}
		\label{eq:MS.int}
	\end{subequations}
	Here, \eqref{eq:MS.1} follows from primal feasibility condition, \eqref{eq:MS.2} follows from \eqref{eq:grad.2}, and \eqref{eq:MS.3} follows from \eqref{eq:CS.1}. If the lower voltage limits are non-binding at all buses at an optimal solution, then \eqref{eq:CS.2} further yields
	$	\ul{\mu}^{v,\star}_k = 0 $
	for each $k$. Then, \eqref{eq:MS.int} implies
	\begin{align*}
	\textrm{MS} 
	= \sum_{k \ell = 1}^m \mu_{k \ell}^\star f_{k\ell}
	+ \sum_{k=1}^n \ol{\mu}^{v,\star}_k \ol{v}_k^2
	\geq 0.
	\end{align*}
	The inequality follows from the nonnegativity of each term in each summand, completing the proof.
\end{proof}

Theorem \ref{thm:2} asserts that payments from uncontrollable assets cover the rents payable to controllable assets, provided lower bounds for voltage constraints do not bind at any bus. This requirement is only \emph{sufficient} for revenue adequacy. In Section \ref{sec:3bus}, we provide an example where MS $>0$ even when the condition is violated, proving that it is not necessary.

\begin{corollary}
	If the voltage lower limits are non-binding at a global optimum of $\Pac$ with a zero duality gap, then SDP-LMPs, together with that global optimum of $\Pac$, define a revenue adequate mechanism.
\end{corollary}
The proof is immediate from Theorems \ref{thm:epigraph}(c) and \ref{thm:2}.
When there is duality gap, the payment scheme outlined here with SDP-LMPs may not provide adequate dispatch-following incentives--a property we study in detail in the next section.

\section{Market Equilibrium Properties of the Prices}
\label{sec:marketEq}
Ideally, a pricing scheme should be such that it is in the best interest of the market participants to follow the SO's dispatch signals. We now study if our candidate pricing mechanisms satisfy such properties. Call a pricing mechanism \emph{individually rational}, if the SO-prescribed dispatch maximizes the profit of a controllable asset, given the prices. That is, a dispatch $(p_k^{G, \star}, q_k^{G,\star})$ is individually rational if it solves
\begin{equation}
\begin{aligned}
& \underset{p_k^G, q_k^G}{\text{maximize}} && \gamma_{k}^{p} p_k^G + \gamma_{k}^{q} q_k^G - c_k(p_k^G, q_k^G),\\
& \text{subject to} && \ul{p}_k \leq p_k^G \leq \ol{p}_k,\ \ \ul{q}_k \leq q_k^G \leq \ol{q}_k,
\end{aligned}
\label{eq:ME.IR}
\end{equation}
given nodal real and reactive power prices $\gamma_k^{p}$ and $\gamma_k^q$. In such an event, a controllable asset has no incentive to deviate from its prescribed dispatch.

Consider a local optimum $(\v{p}^{G,\star}, \v{q}^{G, \star}, \v{V}^\star)$ of $\Pac$. This dispatch is individually rational with AC-LMPs, if $(p_k^{G, \star}, q_k^{G,\star})$ solves \eqref{eq:ME.IR} for all $k$ with $\v{\gamma}^p = \v{\Lambda}^{p,\star}, \v{\gamma}^q = \v{\Lambda}^{q,\star}$. With SDP-LMPs, the same dispatch is individually rational if the same condition holds for \eqref{eq:ME.IR} with $\v{\gamma}^p = \v{\lambda}^{p,\star}, \v{\gamma}^q = \v{\lambda}^{q,\star}$ for all $k$.


A dispatch is said to be \emph{efficient} and clears the market, if it optimally solves $\Pac$. We say a market mechanism supports a \emph{market equilibrium} if the dispatch clears the market and is individually rational, given the vectors of nodal prices.

As our next result will demonstrate, AC-LMPs associated with a locally optimal dispatch of $\Pac$ always provide adequate dispatch following incentives. However, SDP-LMPs coupled with that dispatch may not adequately incentivize all assets to follow the SO instructions. With even a global optimum of $\Pac$, SDP-LMPs may fail to provide such incentives with nonzero duality gap. In these cases, pricing via SDP-LMPs requires the provision of side-payments to controllable assets to deter possible deviations. Despite this critical drawback of SDP-based pricing, we show in the sequel that SDP-LMPs seek to minimize a sum of two nonnegative terms, one of which is the aggregate  side-payments. 

\newcommand{\local}{\textrm{local}}
For a local optimal solution $(\v{p}^{G, \star}, \v{q}^{G, \star}, \v{V}^\star)$ of $\Pac$, define $J^\star_{\AC,\local}$ as the objective of $\Pac$ at that local optimum. Also, define the \emph{lost opportunity costs} associated with SDP-LMPs and that dispatch as
\begin{align}
\textrm{LOC}(\v{\lambda}^{p}, \v{\lambda}^{q}) 
:= 
\sum_{k=1}^{n} \left[ \pi_{k}^{\opt}(\lambda_{k}^p, \lambda_{k}^q)  - \pi_{k}^{\SO}(\lambda_{k}^p, \lambda_{q}^q) \right],
\label{eq:LOC}
\end{align}
where $\pi_{k}^{\opt}$ is the optimal cost of \eqref{eq:ME.IR} with $\v{\gamma}^p = \v{\lambda}^{p,\star}, \v{\gamma}^q = \v{\lambda}^{q,\star}$ and
\begin{align}
\begin{aligned}
\pi_{k}^{\SO}(\lambda_{k}^p, \lambda_{q}^q) 
& := \lambda_{k}^p p_k^{G,\star} + \lambda_{k}^qq_k^{G,\star} - c_k(p_k^{G,\star}, q_k^{G,\star}).
\end{aligned}
\label{eq:LOC.part.def}
\end{align}
That is, given the electricity prices, $\pi_{k}^{\SO}$ denotes the profit of the controllable asset at bus $k$ from following the SO-prescribed dispatch, while $\pi_k^\opt$ is the maximum profit that asset can garner. Finally, with the same local optimum of $\Pac$, define the \emph{product revenue shortfall} as
\begin{align}
\begin{aligned}
\textrm{PRS}(\v{\mu}, \ol{\v{\mu}}^v, \ul{\v{\mu}}^v, \v{U}) 
&  := \v{V}^{\Hsf,\star} \v{U} \v{V^\star} + \sum_{k\ell = 1}^{m} \mu_{k\ell} \left(f_{k\ell} - \v{V}^{{\Hsf},\star} \v{\Phi}_{k \ell} \v{V}^\star \right)\\
& \qquad + \sum_{k=1}^{n} \ol{\mu}^{v}_k \left(\ol{v}_k^2 - | V_k^\star |^2 \right)
+ \sum_{k=1}^{n} \ul{\mu}^{v}_k \left( | V_k^\star |^2 - \ul{v}_k^2\right)
\end{aligned}
\label{eq:PRS.def}
\end{align}
for $\v{\mu} \geq 0, \ol{\v{\mu}}^v \geq 0, \ul{\v{\mu}}^v \geq 0, \v{U} \succeq 0$.

\begin{theorem}
	\label{thm:gap}
	The following assertions hold:
	\begin{enumerate}[label=(\alph*)]
		\item A local optimum of $\Pac$ and its associated AC-LMPs support a market equilibrium. That equilibrium is efficient if the optimum is global.
		
		\item A locally optimal solution of $\Pac$, together with SDP-LMPs may not always support a market equilibrium and 
		\begin{alignat*}{2}
		J_{\AC,\local}^\star - J_\SDP^\star
		& = \underset{\substack{\v{\lambda}^{p}, \v{\lambda}^{q}, \v{\mu}\\ \v{U} \ol{\v{\mu}}^v, \ul{\v{\mu}}^v }}{\text{minimum}} \ \  
		&&   \textrm{LOC}\left(\v{\lambda}_{p}, \v{\lambda}_{q} \right) 
		+ \textrm{PRS}(\v{\mu}, \ol{\v{\mu}}^v, \ul{\v{\mu}}^v, \v{U}),
		\\ 
		& \ \ \ \text{subject to} && 
		\v{U} 
		= \sum_{k=1}^{n} \left( \lambda_k^p \v{\Phi}_k 
		+  \lambda_k^q \v{\Psi}_k  \right)
		+ \sum_{k\ell=1}^{m} \mu_{k\ell} \v{\Phi}_{k\ell} 
		+ \sum_{k=1}^{n} \left( \ol{\mu}^{v}_k - \ul{\mu}^{v}_k \right)\bone_k \bone_k^{\Hsf}, \\
		&&&  \ \v{\mu} \geq 0, \ \ol{\v{\mu}}^v \geq 0, \ \ul{\v{\mu}}^v \geq 0, \ \v{U} \succeq 0.
		\end{alignat*}
		With a global optimal solution of $\Pac$, the above difference equals the duality gap of $\Pac$. 
		
		\item When $\Pac$ has zero duality gap, a globally optimal solution of $\Pac$ together with SDP-LMPs, support a market equilibrium.
	\end{enumerate}
	
\end{theorem}
\begin{proof}
	We prove each part separately.
	
	\noindent $\bullet$ Proof of part (a): The optimization problem in \eqref{eq:ME.IR} is convex with linear inequality constraints, for which KKT optimality conditions are sufficient. Assign Lagrange multipliers $\ol{M}_k^p$, $\ul{M}_k^p$, $\ol{M}_k^q$ and $\ul{M}_k^q$ to the upper and lower, real and reactive power limits in 
	\eqref{eq:ME.IR}. Then, the KKT conditions of \eqref{eq:ME.IR} comprise its feasibility constraints, the dual feasibility constraints $\ol{M}_k^p$, $\ul{M}_k^p$, $\ol{M}_k^q, \ul{M}_k^q \geq 0$, the stationarity conditions
	\begin{align}
	\begin{aligned}
	&  \nabla_{p_k^G}  \left[ c_k(p_k^{G,\star},q_k^{G,\star}) \right]
	- \gamma_k^{p}
	+ \ol{M}_k^{p,\star} - \ul{M}_k^{p,\star} = 0,  \\
	&   \nabla_{q_k^G}  \left[ c_k(p_k^{G,\star},q_k^{G,\star}) \right]
	- \gamma_k^{q}
	+ \ol{M}_k^{q,\star} - \ul{M}_k^{q,\star} = 0,
	\end{aligned}
	\label{eq:kkt.me.ir}
	\end{align}
	and the complementary slackness conditions
	\begin{align}
	\begin{aligned}
	\ul{M}^{q,\star}_k \left(q^{G,\star}_k - \ul{q}^G_k \right) 
	= 
	\ol{M}^{q,\star}_k \left(q^{G,\star}_k - \ol{q}^G_k \right) 
	= 
	\ul{M}^{p,\star}_k \left(p^{G,\star}_k - \ul{p}^G_k \right) 
	= 
	\ol{M}^{p,\star}_k \left(p^{G,\star}_k - \ol{p}^G_k \right) = 0.
	\end{aligned}
	\end{align}
	These KKT conditions define a subset of the KKT conditions of $\Pac$ with $\v{\gamma}_k$'s as $\Lambda_k$'s and $M_k$'s as $\mu_k$'s, proving the result.

	\noindent $\bullet$ Proof of part (b): For a local optimum of $\Pac$, we have
	\begin{gather}
	\begin{gathered}
	p^D_k = p^{G,\star}_k - \v{V}^{\Hsf,\star}  \v{\Phi}_k \v{V}^\star, \ \ q^D_k = q^{G,\star}_k - \v{V}^{\Hsf,\star}  \v{\Psi}_k \v{V}^\star,
	\ \ 
	{J^\star_{\AC,\local}}=\sum_{k=1}^n c_k(p_k^{G,\star}, q_k^{G,\star}).
	\end{gathered}
	\label{eq:nodalDem}
	\end{gather}
	Utilizing these relations in the definition of $\zeta$ in \eqref{eq:zeta.def}, we get
	\begin{align}
	{
		\begin{aligned}
		&\zeta(\v{p}^G, \v{q}^G,  \v{\lambda}^p, \v{\lambda}^q, \v{\mu}, \ol{\v{\mu}}^v, \ul{\v{\mu}}^v) \\
		&= {J^\star_{\AC,\local}} + \sum_{k=1}^{n} \left\lbrack c_k(p^G_k, q^G_k) - \lambda_k^p p^G_k - \lambda_k^q q^G_k \right\rbrack 
		+ \sum_{k=1}^{n} \left[ \lambda_k^p p^{G,\star}_k + \lambda_k^q q^{G,\star}_k - c_k(p^{G,\star}_k, q^{G,\star}_k) \right]
		\\
		& \quad - \sum_{k\ell=1}^{m} \mu_{k\ell}  f_{k\ell} 
		+  \sum_{k=1}^{n} \left(\ul{\mu}^{v}_k \ul{v}_k^2 - \ol{\mu}^{v}_k \ol{v}_k^2\right) 
		- \sum_{k=1}^{n} \left\lbrack \lambda_k^p \v{V}^{\Hsf,\star}  \v{\Phi}_k \v{V}^\star + \lambda_k^q \v{V}^{\Hsf,\star}\v{\Psi}_k \v{V}^\star  \right\rbrack.
		\end{aligned}
	}
	\label{eq:zeta.local}
	\end{align}
	Recall that $\Dcal\Pac = \Dcal\Psdp$ defines the common dual program of $\Pac$ and $\Psdp$. Strong duality holds for $\Psdp$, and hence, $J_\SDP^\star$ is the optimal cost of \eqref{eq:DPac.2}. Utilizing \eqref{eq:zeta.local} in \eqref{eq:DPac.2}, we get
	\begin{alignat}{2}
	{
		\begin{aligned}
		&J^\star_{\AC,\local} - J^\star_{\SDP}
		\\
		& =  - \underset{\substack{\v{\mu},  \ol{\v{\mu}}^v, \ul{\v{\mu}}^v \geq 0 \\
				{\v{U}}, \v{\Lambda}^p, \v{\Lambda}^q}}{\text{maximum}} \ \ 
		\underbrace{\sum_{k=1}^{n} \left[ - \pi^\opt_k(\lambda^p_k, \lambda^q_k) + \pi^\SO_k(\lambda^p_k, \lambda^q_k) \right]}_{:=-\textrm{LOC}\left(\v{\lambda}^p, \v{\lambda}^q \right)} 
		\\ &  \qquad \qquad \qquad
		- \sum_{k\ell=1}^{m} \mu_{k\ell}  f_{k\ell} 
		+  \sum_{k=1}^{n} \left(\ul{\mu}^{v}_k \ul{v}_k^2 - \ol{\mu}^{v}_k \ol{v}_k^2\right) 
		- \sum_{k=1}^{n} \left[ \lambda_k^p \v{V}^{\Hsf,\star}  \v{\Phi}_k \v{V}^\star 
		+\lambda_k^q \v{V}^{\Hsf,\star}\v{\Psi}_k \v{V}^\star \right],\\
		& \qquad\text{subject to} \ \ 
		\v{U} 
		= \sum_{k=1}^{n} \left( \lambda_k^p \v{\Phi}_k 
		+  \lambda_k^q \v{\Psi}_k \right)
		+ \sum_{k\ell=1}^{m} \mu_{k\ell} \v{\Phi}_{k\ell} 
		+ \sum_{k=1}^{n} \left( \ol{\mu}^{v}_k - \ul{\mu}^{v}_k \right)\bone_k \bone_k^{\Hsf},
		\\
		& \qquad \qquad \qquad \ \ \v{\mu} \geq 0, \ \ol{\v{\mu}}^v \geq 0, \ \ul{\v{\mu}}^v \geq 0, \ {\v{U}} \succeq 0.
		\end{aligned}
	}
	\label{eq:AC.local.gap}
	\end{alignat}
	Write the objective function of the above optimization problem as $\eta -\textrm{LOC}\left(\v{\lambda}^p, \v{\lambda}^q \right)$. Then, the expression for $\v{U}$ in the constraint can be used to simplify $\eta$ as
	
	\begin{align}
	\begin{aligned}
	-\eta
	&= \sum_{k\ell=1}^{m} \mu_{k\ell}  f_{k\ell} 
	-  \sum_{k=1}^{n} \left(\ul{\mu}^{v}_k \ul{v}_k^2 - \ol{\mu}^{v}_k \ol{v}_k^2\right) 
	+ \v{V}^{\Hsf,\star} {\v{U}} \v{V}^\star 
	\\ 
	& 
	\ - \sum_{k\ell=1}^{m} \mu_{k\ell} \v{V}^{\Hsf,\star} \v{\Phi}_{k\ell} \v{V}^\star 
	- \sum_{k=1}^{n}  \left( \ol{\mu}^{v}_k - \ul{\mu}^{v}_k \right) \underbrace{\v{V}^{\Hsf,\star} \bone_k \bone_k^{\Hsf}  \v{V}^\star}_{:=|V_k^\star|^2} \\
	& = \textrm{PRS} \left(\v{\mu}, \ol{\v{\mu}}^v, \ul{\v{\mu}}^v, {\v{U}} \right).
	\end{aligned}
	\end{align}

	\noindent $\bullet$ Proof of part (c):
	This follows from combining part (a) and Theorem \ref{thm:epigraph}(c).
\end{proof}


We discuss the implications of Theorem \ref{thm:gap} and contrast it with similar results known for CHPs. Note that our characterization of the cost gap between a locally optimal solution of $\Pac$ and its dual in Theorem \ref{thm:gap}(b) bears a striking resemblance with the duality gap of CHPs derived in \cite{Schiro}. Indeed, the analysis in \cite{Schiro} shows that CHPs seek to minimize the sum of LOC and PRS defined within the context of a unit commitment problem. This parallel between CHPs and SDP-LMPs is not surprising, given that both advocate pricing based on the dual (or the double dual) of the nonconvex market clearing problem, albeit to tackle two different kinds of nonconvexities. However, we point out that Theorem \ref{thm:gap}(b) does \emph{not} enjoy the same interpretation as the duality gap result for CHPs. Since CHPs minimize LOC $+$ PRS, which are individually non-negative, its attempt to reduce LOC can be viewed as a means to mitigate the net out-of-market settlements that the SO must provide the market participants for them to follow the SO-intended dispatch. In unit commitment problems, nodally uniform equilibrium prices may not exist. Consequently, even though the PRS term skews the objective of CHPs from pure LOC reduction, CHPs offer a principled mechanism to approach said reduction.
Our result reveals that SDP-LMPs also minimize LOC $+$ PRS. However, if minimization of LOC is the only goal, AC-LMPs achieve that goal, per Theorem \ref{thm:gap}(a). SDP-LMPs, on the other hand, may fail to eliminate the need for  out-of-market settlements, even though it tries to shrink it as Theorem \ref{thm:gap}(b) reveals.
When duality gap of $\Pac$ vanishes, Theorem \ref{thm:gap}(c) shows that SDP-LMPs obviate the need for such settlements, but only when the dispatch is a global optimum of $\Pac$. In such a case, LOC is provably zero, from Theorem \ref{thm:gap}(b).

\setlength{\tabcolsep}{8pt}
\begin{table}[ht]
	\centering
	\caption{Parameter choices for the experiments  on the three-bus power network.}
	\label{tab:ra}
	{\small
	\begin{tabular}{c c c c c c c c c c c c c} 
		\toprule
		Exp. &  $f$ & $r$ & $x$ & $k$ &
		$p_k^D$ & $q_k^D$ & $\ol{p}_k^G$ & $\ol{q}_k^G$ & $\ul{v}_k^2$ & $\ol{v}_k^2$ & $C^1_k$ & $C^2_k$ \\ [0.5ex] 
		\midrule
		\multicolumn{1}{c}{} & \multicolumn{1}{c}{} & \multicolumn{1}{c}{} & \multicolumn{1}{c}{} & 1 & 0.79 & 0.50 & 2.00 & 0.90 & 0.95 & 0.98 & 10 & 1.0 \\
		\multicolumn{1}{c}{\multirow{-1}{*}{1}} & \multicolumn{1}{c}{\multirow{-1}{*}{0.24}} & \multicolumn{1}{c}{\multirow{-1}{*}{0.01}} & \multicolumn{1}{c}{\multirow{-1}{*}{0.01}} & 2 & 0 & 0 & 1.20 & 0.21 & 0.98 & 1.01 & 10 & 1.0 \\
		\multicolumn{1}{c}{} & \multicolumn{1}{c}{} & \multicolumn{1}{c}{} & \multicolumn{1}{c}{} & 3 & 1.90 & 0 & 2.00 & 2.00 & 0.99 & 1.01 & 10 & 1.0
		\\
		\midrule
		\multicolumn{1}{c}{} & \multicolumn{1}{c}{} & \multicolumn{1}{c}{} & \multicolumn{1}{c}{} & 1 & 0.79 & 0.10 & 2.00 & 0.90 & 0.95 & 1.05 & 10 & 1.0 \\
		\multicolumn{1}{c}{\multirow{-1}{*}{2}} & \multicolumn{1}{c}{\multirow{-1}{*}{0.20}} & \multicolumn{1}{c}{\multirow{-1}{*}{0.01}} & \multicolumn{1}{c}{\multirow{-1}{*}{0.01}} & 2 & 0 & 0 & 1.20 & 0.21 & 0.98 & 1.01 & 10 & 1.0 \\
		\multicolumn{1}{c}{} & \multicolumn{1}{c}{} & \multicolumn{1}{c}{} & \multicolumn{1}{c}{} & 3 & 2.00 & 0 & 2.00 & 2.00 & 0.95 & 1.01 & 10 & 1.0\\ \midrule
		\multicolumn{1}{c}{} & \multicolumn{1}{c}{} & \multicolumn{1}{c}{} & \multicolumn{1}{c}{} & 1 & 0.79 & 0.50 & 2.00 & 0.90 & 1.01 & 1.05 & 10 & 1.0 \\
		\multicolumn{1}{c}{\multirow{-1}{*}{3}} & \multicolumn{1}{c}{\multirow{-1}{*}{0.40}} & \multicolumn{1}{c}{\multirow{-1}{*}{0.01}} & \multicolumn{1}{c}{\multirow{-1}{*}{0.01}} & 2 & 0 & 0 & 1.20 & 0.21 & 0.98 & 1.01 & 10 & 1.0 \\
		\multicolumn{1}{c}{} & \multicolumn{1}{c}{} & \multicolumn{1}{c}{} & \multicolumn{1}{c}{} & 3 & 2.00 & 0 & 2.00 & 2.00 & 0.99 & 1.01 & 10 & 1.0 \\
		\midrule
\multicolumn{1}{c}{} & \multicolumn{1}{c}{} & \multicolumn{1}{c}{} & \multicolumn{1}{c}{} & 1 & 1.10 & 1.00 & 1.00 & 2.00 & 0.98 & 1.01 & 10 & 0.1 \\
\multicolumn{1}{c}{\multirow{-1}{*}{4}} & \multicolumn{1}{c}{\multirow{-1}{*}{0.90}} & \multicolumn{1}{c}{\multirow{-1}{*}{0.03}} & \multicolumn{1}{c}{\multirow{-1}{*}{0.75}} & 2 & 1.10 & 1.00 & 3.00 & 2.00 & 0.99 & 1.01 & 1 & 0.1 \\
\multicolumn{1}{c}{} & \multicolumn{1}{c}{} & \multicolumn{1}{c}{} & \multicolumn{1}{c}{} & 3 & 0.95 & 1.00 & 0 & 2.00 & 0.95 & 1.02 & 0 & 0 \\	
 \bottomrule	
\end{tabular}
}
\end{table}
\begin{table}[h]
	\centering
	\caption{Outputs of the experiments on the three-bus power network.}
	\label{tab:ra.output}
	{\small
	\begin{tabular}{c c c c c c c c c c c c} 
		\toprule
		Exp. & $k$ &
		$p_{k,\SDP}^{G,\star}$ & $q_{k,\SDP}^{G,\star}$ & $p_{k,\AC}^{G,\star}$ & $q_{k,\AC}^{G,\star}$ & ${\lambda}_k^{p,\star}$ & ${\lambda}_k^{q,\star}$ & ${\Lambda}_k^{p,\star}$ & ${\Lambda}_k^{q,\star}$ & $|V_k^\star|^2$ & $\textrm{MS}_{\AC}$ \\ [0.5ex] 
		\midrule
		\multicolumn{1}{c}{} & 1 & 0.39 & 0 & 0.39 & 0 & 10.77 & -4.33 & 10.77 & -4.33 & 0.98 & \multicolumn{1}{c}{}\\
		\multicolumn{1}{c}{\multirow{-1}{*}{1}} & 2 & 0.31 & 0 & 0.31 & 0 & 10.63 & -2.16 & 10.63 & -2.16 & 0.99 & \multicolumn{1}{c}{\multirow{-1}{*}{-2.44}}\\
		\multicolumn{1}{c}{} & 3 & 1.99 & 0.50 & 1.99 & 0.50 & 13.99 & 0 & 13.99 & 0 & 0.99 & \multicolumn{1}{c}{}
		\\
		\midrule
		\multicolumn{1}{c}{} & 1 & 0.92 & 0.10 & 0.92 & 0.10 & 11.85 & 0 & 11.85 & 0 & 1.01 & \multicolumn{1}{c}{}\\
		\multicolumn{1}{c}{\multirow{-1}{*}{2}} &  2 & 0.23 & 0 & 0.23 & 0 & 10.47 & 0 & 10.47 & 0 & 1.01 &\multicolumn{1}{c}{\multirow{-1}{*}{0.83}}\\
		\multicolumn{1}{c}{} & 3 & 1.63 & 0 & 1.63 & 0 & 13.27 & 0 & 13.27 & 0 & 1.01 & \multicolumn{1}{c}{}\\ 
		\midrule
		\multicolumn{1}{c}{}  & 1 & 1.19 & 0.50 & 1.19 & 0.50 & 12.38 & 0 & 12.38 & 0 & 1.01 & \multicolumn{1}{c}{}\\
		\multicolumn{1}{c}{\multirow{-1}{*}{3}} & 2 & 0.40 & 0 & 0.40 & 0 & 10.80 & -1.09 & 10.80 & -1.09 & 1.01 & \multicolumn{1}{c}{\multirow{-1}{*}{0.62}}\\\multicolumn{1}{c}{} & 3 & 1.20 & 0 & 1.20 & 0 & 12.41 & -0.55 & 12.41 & -0.55 & 1.00 & \multicolumn{1}{c}{}\\ 
		\midrule
		\multicolumn{1}{c}{}  & 1 & 0.31 & 1.44 & 0.97 & 1.09 & 10.06 & 0 & 10.20 & 0 & 1.01 & \multicolumn{1}{c}{}\\
		\multicolumn{1}{c}{\multirow{-1}{*}{4}} & 2 & 2.90 & 1.70 & 2.20 & 1.20 & 1.58 & 0 & 1.44 & 0 & 1.01 & \multicolumn{1}{c}{\multirow{-1}{*}{17.55}}\\\multicolumn{1}{c}{} & 3 & 0 & 1.37 & 0 & 1.26 & 11.52 & 0 & 18.78 & 0 & 1.02 & \multicolumn{1}{c}{}\\
 \bottomrule	
	\end{tabular}
	}
\end{table}


\section{A Three-Bus Network Example}
\label{sec:3bus}


We illustrate our theoretical results through a 3-bus power network with generators and demands at each bus, that are connected via three lines with identical parameters $f$, $r$, $x$. Assume quadratic generator costs of the form $c_k(p_k^G, q_k^G)  := C^2_k \left( p_k^G \right)^2 + C^1_k p_k^G$ at each bus $k$. Set $\ul{\v{p}}^G = \ul{\v{q}}^G=0$ throughout.
$\Psdp$ is solved using CVX 2.2, a package for specifying and solving convex programs (see \cite{cvx,gb08}), in MATLAB R2020b with Mosek 9.1.9 as the solver, while $\Pac$ is solved using Matpower 7.1 (see \cite{zimmerman2010matpower}).  We present results from four experiments on this network. Parameter choices and outputs are listed in Tables \ref{tab:ra} and \ref{tab:ra.output}, respectively. The code is available at \url{https://github.com/Mariola-Nd/RLMP.git}.

\subsection{Experiments with Zero Duality Gap}
In the first three experiments, we obtain $\rank \v{W}^\star = 1$ from $\Psdp$. Thus, duality gap is zero. Also, Matpower discovers a certifiably global optimal solution with the same cost as $\Psdp$. Not surprisingly, AC-LMPs and SDP-LMPs coincide, i.e., we obtain $\v{\Lambda}^{p,\star} = \v{\lambda}^{p,\star}$ and $\v{\Lambda}^{q,\star} = \v{\lambda}^{q,\star}$ as Theorem \ref{thm:epigraph}(c) dictates. In the first experiment, note that the voltage lower limit at bus 3 binds. This is an example where the sufficient condition in Theorem \ref{thm:2} for revenue adequacy is violated and we do obtain MS $< 0$. 
In the second experiment, voltage lower limits do not bind at any bus. Indeed, we obtain a non-negative MS, as Theorem \ref{thm:2} dictates. For the third case, the lower limit on voltage magnitude binds at bus 1. Yet, we obtain MS $> 0$, indicating that our criterion identified in Theorem \ref{thm:2} for revenue adequacy is \emph{sufficient but not necessary}.

Both prices support an efficient market equilibrium. Specifically, the dispatch of each generator maximizes its profits, given the prices. To explicitly illustrate this, consider the dispatch of generator at bus 1 in the second experiment. We verified using CVX that indeed $(p_1^{G,\star}, q_1^{G, \star}) = (0.92, 0.10)$ yields the maximum profit attainable by generator 1 within $[0, 2.00] \times [0,0.90]$, given the prices $ \lambda^{p, \star}_1 = \Lambda^{p, \star}_1 = 11.85$ and $\lambda^{q, \star}_1 = \Lambda^{q, \star}_1 = 0$. Similar conclusions hold for all generators.

\subsection{Experiment with Possibly Non-Zero Duality Gap}
With the parameters chosen for experiment 4, Matpower converges to a solution of $\Pac$ with cost $\$12.51$/MWh while $\Psdp$ finds a solution with the lower cost of $\$6.86$/MWh. The solution of $\Pac$ is a local optimal solution, according to Matpower; global optimality remains difficult to certify. $\Psdp$ returns a solution with $\rank \v{W}^\star = 2$.
Thus, $\Pac$ has a possibly non-zero duality gap. From Table \ref{tab:ra.output}, it is evident that the dispatch and the prices from $\Pac$ and $\Psdp$ are different. 
With AC-LMPs, the locally optimal dispatch of $\Pac$ gives a positive MS of 17.55, where the lower limits on the voltage magnitudes do not bind at any bus, as we expect from Theorem \ref{thm:2}. 
If market participants are compensated via the AC-LMPs $(\v{\Lambda}^{p,\star}, \v{\Lambda}^{q,\star})$ for producing the optimal solution $(\v{p}^{G,\star}, \v{q}^{G,\star})$ from $\Pac$, no side-payments are necessary. To illustrate this fact further, consider the dispatch of generator at bus 2. Indeed, $(p_2^{G,\star}, q_2^{G, \star}) = (2.20, 1.20)$ yields the maximum profit attainable by generator 2 within $[0, 3.00] \times [0, 2.00]$, given the prices $\Lambda^{p, \star}_2 = 1.44$ and $\Lambda^{q, \star}_2 = 0$.
On the other hand, pricing via SDP-LMPs $(\v{\lambda}^{p,\star}, \v{\lambda}^{q,\star})$ would not adequately incentivize participants to follow the dispatch solution of $\Pac$. Specifically, given $\lambda^{p,\star}_2 = 1.58, \lambda^{q,\star}_2 = 0$, that generator's profit becomes $\pi_{2}^{\SO}({\lambda}_{2}^{p,\star}, {\lambda}_{2}^{q,\star}) = 0.79$ with the dispatch from $\Pac$. The maximum attainable profit of that generator with these prices, however, is $\pi^\opt_2({\lambda}_{2}^{p,\star}, {\lambda}_{2}^{q,\star}) = 0.84$ with a production of $\left(p^G_2, q^G_2\right) = (2.90, 0)$. A side-payment 
is necessary for the generator to follow SO's dispatch signal.



\section{SDP-LMPs For Distribution Networks}
\label{sec:DLMP}
The aim to harness flexibility offered by distributed energy resources (DERs) at the grid-edge has motivated research in defining appropriate price signals for compensating energy transactions in distribution networks, e.g., see  \cite{Ntakou}, \cite{PapavasiliouDLMPs} and \cite{Yuan2018}.
Suggested distribution LMPs (DLMPs) aim to reflect the locational value of DERs as discussed in \cite{Li} and \cite{Huang}.
We argue that prices from $\Psdp$ become the second-order cone programming (SOCP) based DLMPs in \cite{PapavasiliouDLMPs, Caramanis} over acyclic distribution grids.

Represent an acyclic distribution network as a directed graph over the set $\Nset$ of buses  with \emph{directed} edges in $\Eset$. The directions can be arbitrarily chosen. Ignore shunt admittances for simplicity. 
Denote by $k \to \ell$, a directed edge from bus $k$ to bus $\ell$. Define ${1}/{y_{k\ell}} :=r_{k \ell} + \bm{i} x_{k\ell}$
as the impedance of line $k \to \ell$. Using this notation, consider the following optimization program.
\begin{subequations}
	\begin{alignat}{2}
	&\Psocp \ : \notag
	\\
	& \ {\text{minimize}}  \  \sum_{k=1}^n c_k(p^G_k, q^G_k), \notag \\ 
	& \ \text{subject to}  \notag \ \
	\\
	& \   p_k^G - p_k^D = \sum_{\ell': k \to \ell'} P_{k\ell'} 
	-\sum_{\ell': \ell' \to k} \left( P_{\ell' k} - r_{\ell' k} J_{\ell' k}\right),
	\label{eq:SOCP.pV}
	\\ 
	& \   q_k^G - q_k^D =  \sum_{\ell': k \to \ell'} Q_{k\ell'} 
	-\sum_{\ell': \ell' \to k} \left( Q_{\ell' k} - x_{k\ell'} J_{\ell' k}\right),
	\label{eq:SOCP.qV}
	\\
	& \   P_{k\ell} \leq f_{k \ell}, 
	\ 
	r_{k \ell} J_{k \ell} - P_{k \ell} \leq f_{k \ell}, \label{eq:SOCP.fLim}
	\\
	& \  \ul{p}_k \leq p_k^G \leq \ol{p}_k, 
	\ul{q}_k \leq q_k^G \leq \ol{q}_k, 
	\label{eq:SOCP.pqG}
	\\
	& \  \ul{v}_k^2 \leq w_k \leq \ol{v}_k^2,  
	\label{eq:SOCP.v}
	\\			
	& \  w_\ell = w_k - 2 (P_{k\ell} r_{k \ell} + Q_{k \ell} x_{k\ell}) 
	+ (r_{k\ell}^2 + x_{k\ell}^2) J_{k\ell},
	\label{eq:SOCP.vPQ}
	\\
	& \  P_{k\ell}^2 + Q_{k\ell}^2 \leq J_{k \ell} w_k
	\label{eq:SOCP.socp}
	\\
	& \   \text{for } k \in \Nset, \ k \to \ell \in \Eset 
	\notag
	\end{alignat}
	\label{eq:Psocp}
\end{subequations}
over the variables $\v{p}^G, \v{q}^G, \v{w}, \v{P}, \v{Q}, \v{J}$. All constraints in the above problem are linear except \eqref{eq:SOCP.socp} that is a second-order cone constraint. In fact, the inequality in \eqref{eq:SOCP.socp} replaced by an equality amounts to a reformulation of $\Pac$. The inequality potentially expands the feasible set of $\Pac$, making $\Psocp$ a convex relaxation of $\Pac$. 
When solved with an equality in \eqref{eq:SOCP.socp}, the variables $P_{k\ell}$ and $Q_{k \ell}$ denote the sending-end real and reactive powers from bus $k$ towards bus $\ell$. Then, $J_{k\ell}$ becomes the squared current magnitude on that line. And, $w_k$ equals the squared voltage magnitude at bus $k$.  The SOCP-based relaxation presented above utilizes the so-called ``branch flow model'' of Kirchhoff's laws over a distribution network, and has been extensively analyzed in \cite{Gan}, \cite{NaLi} and \cite{Farivar}. Constraints \eqref{eq:SOCP.pV} and \eqref{eq:SOCP.qV} encode nodal real and reactive power balance, respectively. Inequalities in \eqref{eq:SOCP.fLim} enforce limits on distribution line flows. The capacities of controllable assets are given by \eqref{eq:SOCP.pqG} and voltage limits are encoded in \eqref{eq:SOCP.v}. The equality in \eqref{eq:SOCP.vPQ} relates the power flows on lines with squared voltage magnitudes across the lines. 

Associate Lagrange multipliers $\rho_{k}^p$ and $\rho_{k}^q$ with the power balance constraints \eqref{eq:SOCP.pV} and \eqref{eq:SOCP.qV}, respectively. Call their respective collections across the network as $\v{\rho}^{p}$ and $\v{\rho}^{q}$.
\begin{definition}[SOCP-DLMPs] The optimal Lagrange multipliers $\v{\rho}^{p,\star}$ and $\v{\rho}^{q,\star}$ for $\Psocp$ define the SOCP distribution locational marginal prices (SOCP-DLMPs) for real and reactive powers, respectively.
\end{definition}
One can correspondingly consider $\Psdp$ for the same radial network and derive SDP-LMPs $\left(\v{\lambda}^{p,\star}, \v{\lambda}^{q,\star}\right)$ as optimal dual multipliers for $\Psdp$. We now establish a relationship between SOCP-DLMPs and SDP-LMPs.

\begin{theorem}
	\label{thm:DLMP}
	For a radial power network, SDP-LMPs are SOCP-DLMPs and vice-versa.
\end{theorem}
\begin{proof}
	$\Psdp$ can be written as 
	\begin{alignat}{2}
	\begin{aligned}
	& {\text{minimize}} && \ \sum_{k=1}^n c_k(p^G_k, q^G_k), 
	\\
	& \text{subject to} 
	&& \  p^G_k - p^D_k = \sum_{\ell' : k \to \ell'} p_{k\ell'} + \sum_{\ell' : \ell' \to k} p_{k \ell'}, \\
	&&& \   q^G_k - q^D_k = \sum_{\ell' : k \to \ell'} q_{k\ell'} + \sum_{\ell' : k \to \ell'}q_{k \ell'}, \\
	&&& \  p_{k\ell} \leq f_{k \ell}, \ p_{\ell k} \leq f_{k\ell},
	\ \eqref{eq:SOCP.pqG}, \ \eqref{eq:SOCP.v}
	\\
	&&& \ \text{for } k \in \Nset, \ k\to\ell \in \Eset, 
	\\
	&&& \  \left(\v{p}_e, \v{p}_{e'}, \v{q}_e, \v{q}_{e'}, \v{w} \right) \in \Wset.
	\end{aligned}
	\label{eq:Psdp.1}
	\end{alignat}
	The vectors $\v{p}_e$ and $\v{p}_{e'}$ collect $p_{k \ell}$  and $p_{\ell k}$ for $k\to \ell \in \Eset$, respectively. Similarly, $\v{q}_e$ and $\v{q}_{e'}$ collect all $q_{k \ell}$  and $q_{\ell k}$ for $k\to \ell \in \Eset$, respectively. The set $\Wset$ is defined in \eqref{eq:W.def}. From \cite[Theorem 6]{bose2015equivalent}, we have 
	$\Xset = \Wset,$
	where $\Xset$ is as defined in \eqref{eq:X.def}. Replacing $\Wset$ by $\Xset$, \eqref{eq:Psdp.1} becomes $\Psocp$. Thus, the set of optimal dual multipliers of the power balance constraints in \eqref{eq:Psdp.1} and $\Psocp$ coincide, completing the proof.
\end{proof}
SDP-LMPs restricted to radial networks coincide with SOCP-DLMPs proposed in \cite{PapavasiliouDLMPs, Caramanis}, according to Theorem \ref{thm:DLMP}. The rest of the results in this paper on SDP-LMPs then characterize properties of SOCP-DLMPs and provide the economic rationale behind using these DLMPs to compensate DERs. In particular, when the SOCP relaxation is exact, these prices are AC-LMPs that support an efficient market equilibrium and ensure revenue adequacy (with non-binding voltage lower limits). By exact, we mean that \eqref{eq:SOCP.socp} is met with an equality at an optimum of $\Psocp$. When the relaxation is not exact, these prices seek to minimize side-payments for DERs to follow a prescribed dispatch signal.

\begin{figure}[ht]
	\hrule
	\begin{align}
		\begin{aligned}
			\Wset :=  \{ \left(\v{p}_e, \v{p}_{e'}, \v{q}_e, \v{q}_{e'}, \v{w} \right)  \ \vert \
			&  p_{k \ell} = \trace(\v{\Phi}_{k\ell} \v{W}), 
			\ p_{\ell k} = \trace(\v{\Phi}_{\ell k} \v{W}), \
			q_{k \ell} = \trace(\v{\Psi}_{k\ell} \v{W}), 
			\ q_{\ell k} = \trace(\v{\Psi}_{\ell k} \v{W}), \\
			&  w_k = \trace(\bone_k \bone_k^\mathsf{T} \v{W}) 
			\text{ for some } \v{W} \succeq 0 \
			\text{for } k \in \Nset, k\to \ell \in \Eset\}.
		\end{aligned}
		\label{eq:W.def}
	\end{align}
	\begin{align}
		\begin{aligned}
			\Xset := \{ \left(\v{p}_e, \v{p}_{e'}, \v{q}_e, \v{q}_{e'}, \v{w} \right)  \ \vert \
			&   p_{k\ell} = P_{k \ell}, \  p_{\ell k} = r_{k \ell} J_{k \ell} - P_{k\ell}, \
			q_{k\ell} = Q_{k \ell}, \  q_{\ell k} = x_{k \ell} J_{k \ell} - Q_{k\ell}, \\
			&  w_\ell = w_k - 2 (P_{k\ell} r_{k \ell} + Q_{k \ell} x_{k\ell})
			+ (r_{k\ell}^2 + x_{k\ell}^2) J_{k\ell}, \
			P_{k\ell}^2 + Q_{k\ell}^2 \leq J_{k \ell} w_k, w_k \geq 0 \\ 
			&  \text{for some } \v{P}, \v{Q}, \v{J} \text{ for } k \in \Nset, k\to \ell \in \Eset
			\}.
		\end{aligned}
		\label{eq:X.def}
	\end{align}
	\hrule
\end{figure}

We computed SOCP-DLMPs on a 15-bus radial network from \cite{PapavasiliouDLMPs} with the  modifications $p_{11}^D = 0.250$ and $q_{11}^D = 0.073$.\footnote{These figures are adapted from our preliminary work in \cite{winnickiISGT}.} The SOCP-DLMPs are portrayed as heat-maps in Figures \ref{fig:case2.p}, \ref{fig:case2.q}. Figures \ref{fig:case1.p}, \ref{fig:case1.q} draws the prices upon increasing power demands at bus 11.  Compared to Figures \ref{fig:case2.p}, \ref{fig:case2.q}, the results in Figures \ref{fig:case1.p}, \ref{fig:case1.q} reveal that real power demands substantially affect the real power prices. In 
Figures \ref{fig:comparison.p},\ref{fig:comparison.q}, we plot the prices upon altering the voltage limits at various buses. These outcomes, when compared to Figures \ref{fig:case2.p}, \ref{fig:case2.q}, show that voltage limits significantly impact the reactive power prices. One expects such a behavior, given the nature of the coupling between reactive power injections and voltage magnitudes in the power flow equations. In all these experiments, the relaxation was found to be exact. Thus, the dispatch from $\Psocp$ solves $\Pac$ and the SOCP-DLMPs are also AC-LMPs. The plots illustrate the locational nature of these prices. For all experiments, we obtained MS$\geq 0$.

\begin{figure}[ht]
        \centering
        \vspace{-16pt} 
        \subfloat[P.1][$\v{\rho}^{p, \star}$]{\includegraphics[trim=30pt 70pt 30pt 50pt ,clip,width=0.30\textwidth]{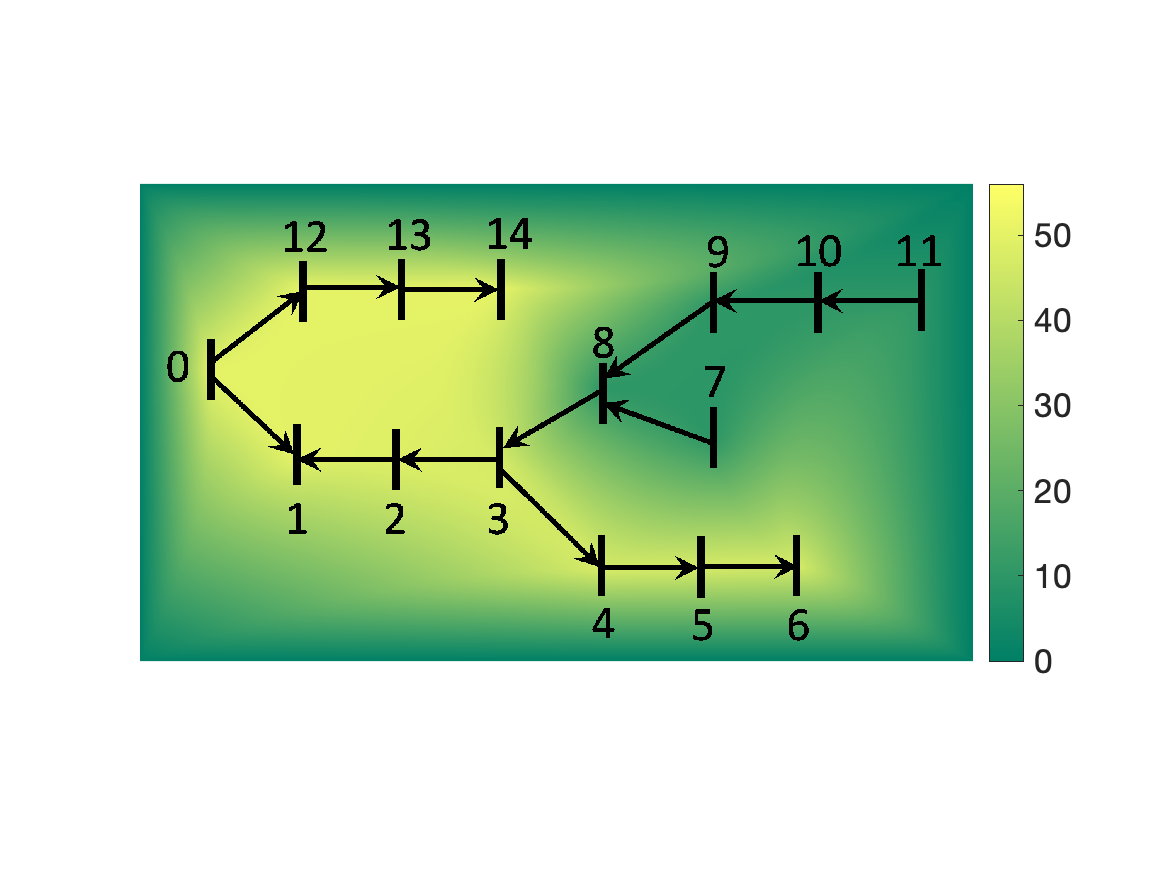} \label{fig:case2.p}}
        \subfloat[Q.1][$\v{\rho}^{q, \star}$]{\includegraphics[
       trim=30pt 70pt 30pt 50pt ,clip,width=0.30\textwidth]{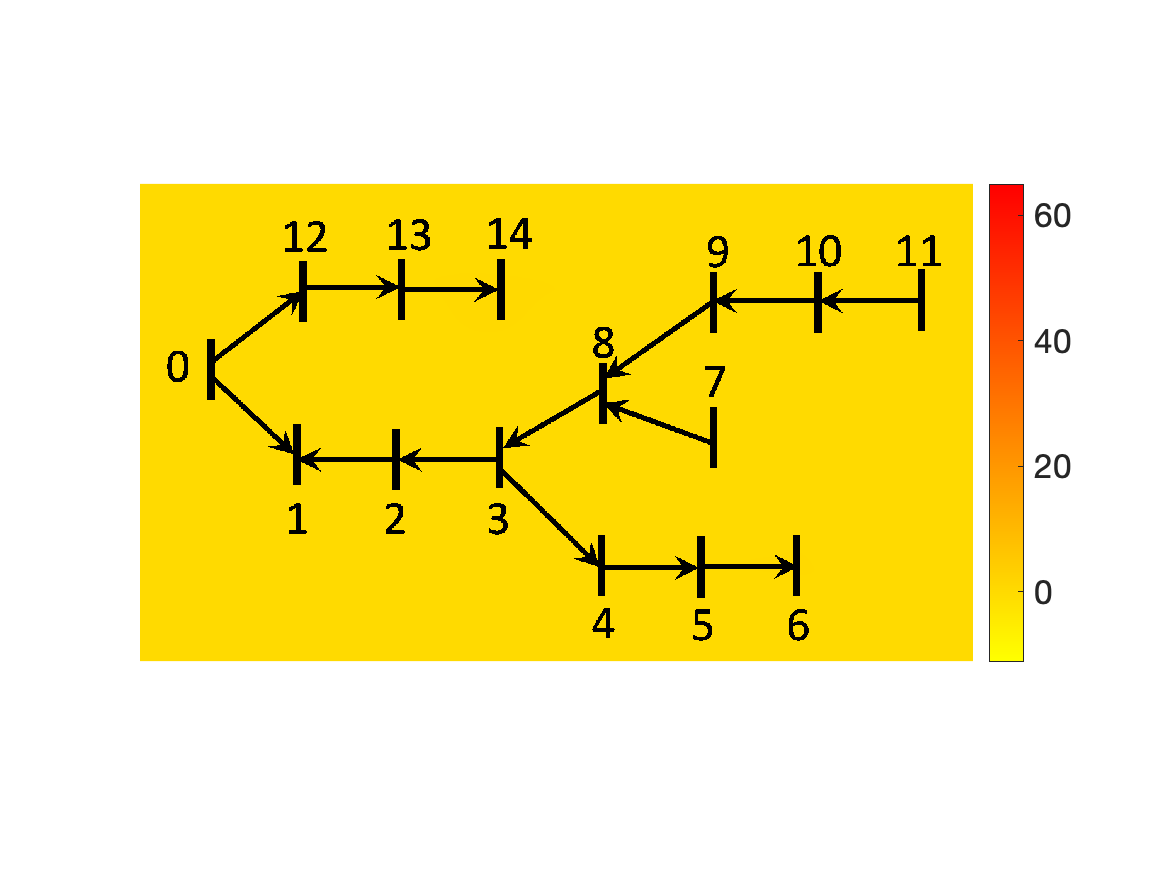} \label{fig:case2.q}}
\\          
        \vspace{-10pt} 
        \subfloat[P.2][$\v{\rho}^{p, \star}$]{\includegraphics[
        trim=30pt 70pt 30pt 50pt ,clip,width=0.30\textwidth]{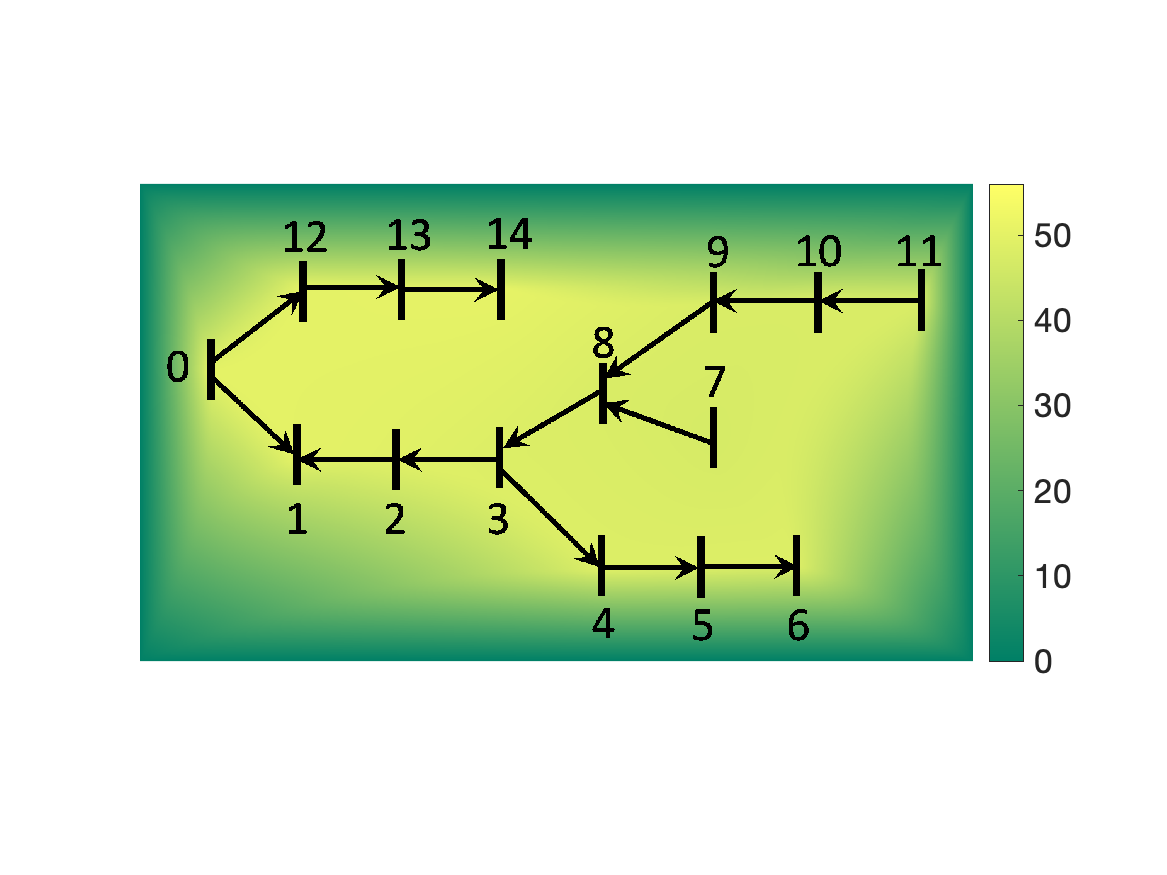} \label{fig:case1.p}}
        \subfloat[Q.2][$\v{\rho}^{q, \star}$]{\includegraphics[
        trim=30pt 70pt 30pt 50pt ,clip,width=0.30\textwidth]{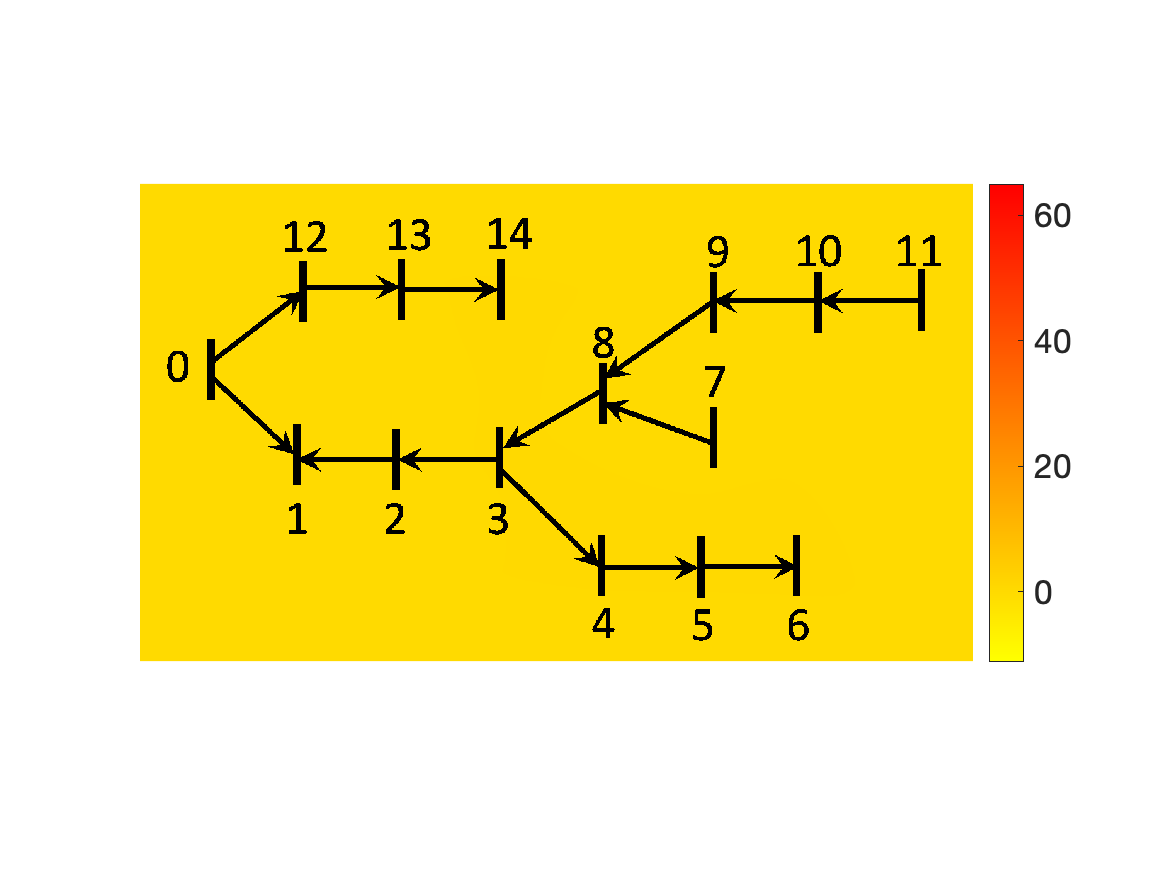} \label{fig:case1.q}}
\\       
        \vspace{-10pt} 
  	    \subfloat[P.3][$\v{\rho}^{p, \star}$]{\includegraphics[
        trim=30pt 70pt 30pt 50pt ,clip,width=0.30\textwidth]{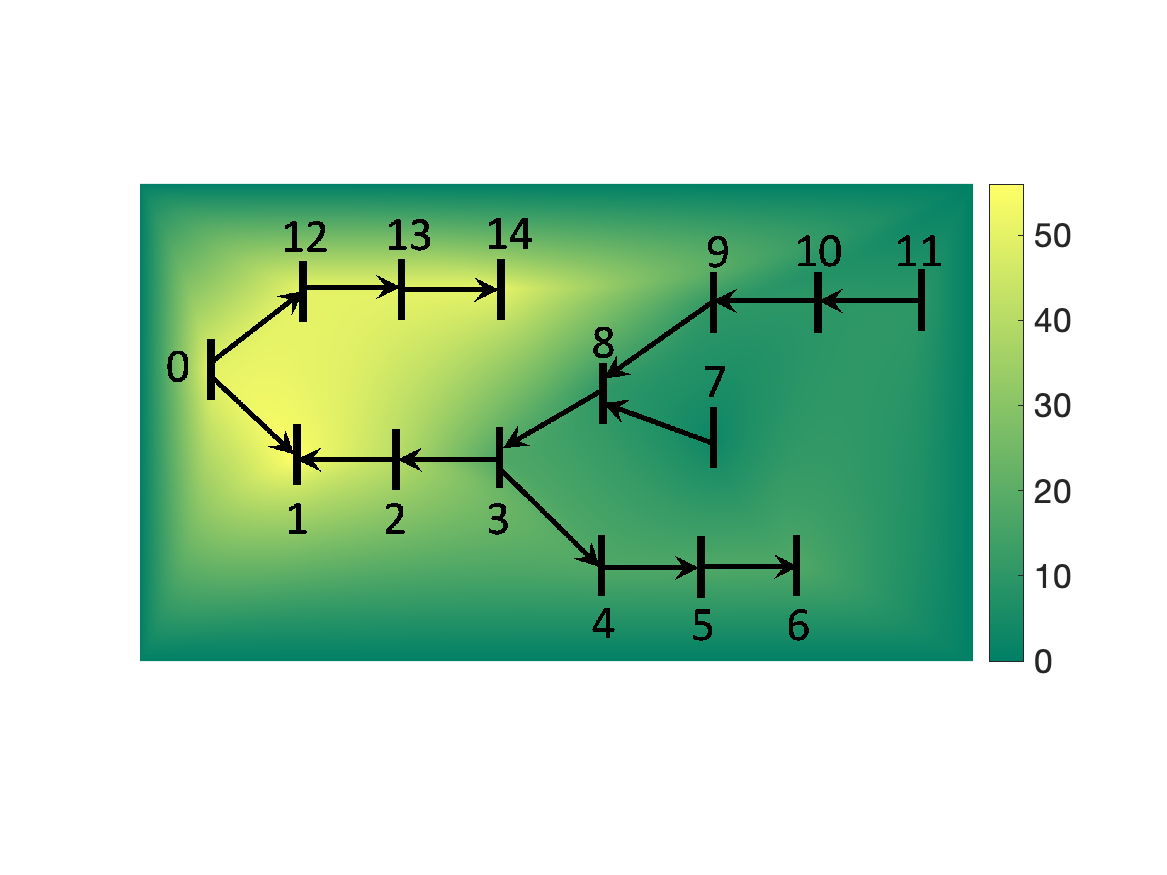} \label{fig:comparison.p}}
        \subfloat[Q.3][$\v{\rho}^{q, \star}$]{\includegraphics[
        trim=30pt 70pt 30pt 50pt ,clip,width=0.30\textwidth]{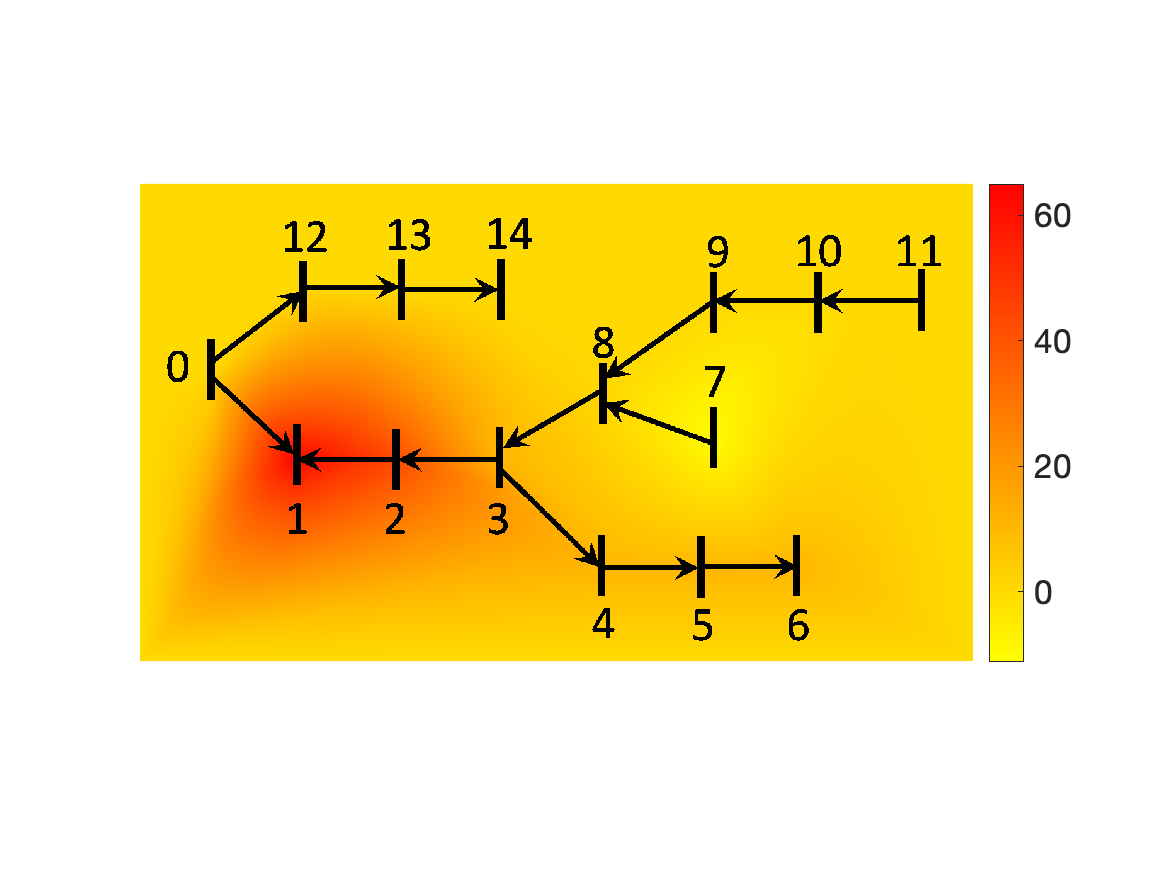} \label{fig:comparison.q}}
    
    \caption{Plots (a), (b) show heat-maps of SOCP-DLMPs (that equal AC-LMPs) on the 15-bus radial network adopted from \cite{PapavasiliouDLMPs}. Plots (c), (d) are derived with $p_{11}^D = 0.350$, and (e), (f) with $\ol{v}_i^2 = 1.05, i=0, \ldots, 10, \ \ul{v}_1^2 = 1$. Arrows indicate the edge directions we considered in $\Psocp$.}               
    \label{fig:heat}
\end{figure}

\section{Conclusions and Future Directions}
\label{sec:conc}
In this paper, we analyzed two candidate pricing mechanisms for market clearing with AC power flow. One set of prices were derived from multipliers that support a KKT system with locally optimal dispatch solutions. The other set of prices were derived from the SDP relaxation of the economic dispatch problems. We established several results that compared these two prices. With zero duality gap, the prices behave similarly as long as the dispatch solution obtained is indeed globally optimal. Otherwise, they can behave differently. SDP-based prices are defined from the Lagrangian dual of the nonconvex market clearing problem and in that respect, bear similarities to convex hull prices defined to tackle nonconvexities in cost structures. Our work shows that while their origins are indeed similar, there are important differences between the two. We also analyzed electricity market-relevant properties such as  revenue adequacy and market equilibrium for the two pricing mechanisms. For transmission networks, these results complement the properties of LMPs derived from linearized power flow equations. When applied to the distribution networks, these results provide new insights into the properties of proposed DLMPs. 

We are interested in two directions for future research. First, we aim to study how our analysis can be generalized to also handle nonconvexities that arise from cost structures. That is, we want to study price formation when we consider commitment decisions and startup costs together with AC power flow equations in market clearing. Second, we want to pursue extensions of our analysis to the stochastic setting that explicitly accounts for uncertainties in renewable supply. These two challenging directions will allow us to better understand pricing in electricity market environments without having to rely on the theory of LMPs that are typically derived from a deterministic economic dispatch problem with linearized power flow equations.






\bibliographystyle{alpha}
\bibliography{myrefs}

\appendix

\begin{lemma}
	\label{lemma:param.fenchel}
	For $f, g:\Xset \to \Rset$, ${a} \in \Rset$ and $\Aset \subseteq \Rset$, define
	\begin{align}
	\begin{aligned}
	J^\star({a}) 
	:= \inf_{\v{x} \in \Xset} f(\v{x}), \text{ subject to } {a} - {g}(\v{x}) \in \Aset.
	\end{aligned}
	\end{align}
	Then, for any ${\xi} \in \Rset$, we have
	\begin{align}
	\begin{aligned}
	J^{\star, \fc}({\xi})
	&=\delta^{\fc}_\Aset({\xi}) - \inf_{\v{x} \in \Xset} \left\{ f(\v{x}) - \xi {g}(\v{x}) \right\}.
	\end{aligned}
	\end{align}
\end{lemma}
\begin{proof} Using the definition of $\delta_\Aset$ in \eqref{eq:supp.f}, we have
	\begin{align}
	\begin{aligned}
	J^{\star, \fc}({\xi})
	&= \sup_{a} \left\{ \xi a - \inf_{\v{x} \in \Xset} \left\{ f(\v{x}) + \delta_\Aset(a - g(\v{x})) \right\} \right\}
	\\
	&= \sup_{\v{x} \in \Xset} \left\{ -f(\v{x}) + \sup_{a} \left\{ \xi a - \delta_\Aset(a-g(\v{x})) \right\} \right\}
	\\
	&= \sup_{x \in \Xset} \left\{ -f(\v{x}) + \xi g(\v{x}) \right.
	\\
	& \qquad \qquad + \left. \sup_{a} \left\{ \xi (a - g(\v{x})) - \delta_\Aset(a-g(\v{x})) \right\} \right\}
	\\
	&= \sup_{\v{x} \in \Xset} \left\{ -f(\v{x}) + \xi g(\v{x}) + \delta^{\fc}_\Aset(\xi) \right\}.
	\end{aligned}
	\end{align}
	The result then follows from the fact that $\delta^{\fc}_\Aset(\xi)$ does not depend on $\v{x}$.
\end{proof}

\end{document}